\documentclass[english,round, 10pt]{article}
\usepackage[T1]{fontenc}
\usepackage[latin9]{inputenc}
\usepackage[a4paper]{geometry}
\usepackage{color}
\usepackage{babel}
\usepackage{float}
\usepackage{units}
\usepackage{mathtools}
\usepackage{enumitem}
\usepackage{amsmath}
\usepackage{amsthm}
\usepackage{amssymb}
\usepackage{graphicx}
\usepackage[authoryear]{natbib}
\usepackage{microtype}
\usepackage[unicode=true,
 bookmarks=true,bookmarksnumbered=false,bookmarksopen=false,
 breaklinks=false,pdfborder={0 0 0},pdfborderstyle={},backref=page,colorlinks=true]
 {hyperref}
\hypersetup{pdftitle={Quantification of Risk in Classical Models of Finance},
 pdfauthor={Ruben Schlotter},
 citecolor=darkgray, linkcolor= blue, urlcolor= magenta}

\makeatletter
\theoremstyle{plain}
\newtheorem{thm}{\protect\theoremname}
\theoremstyle{definition}
\newtheorem{defn}[thm]{\protect\definitionname}
\theoremstyle{remark}
\newtheorem{rem}[thm]{\protect\remarkname}
\theoremstyle{plain}
\newtheorem{lem}[thm]{\protect\lemmaname}
\theoremstyle{plain}
\newtheorem{prop}[thm]{\protect\propositionname}

\RequirePackage[l2tabu, orthodox]{nag} 

\usepackage{tikz, pgfplots}
\usetikzlibrary{arrows, calc, patterns}

\usepackage{dsfont}  		
\usepackage{newtxtext,newtxmath}	
\usepackage{bbm}
\usepackage{doi}
\usepackage{algorithm}
\usepackage{algpseudocode} 
\def\algbackskip{\hskip-\ALG@thistlm}
\usepackage{tikz}
\usetikzlibrary{matrix}

\DeclareMathOperator{\E}{{\mathds E}}	
\DeclareMathOperator{\SD}{\mathsf{SD}}

\DeclareMathOperator{\nd}{\mathsf{d\kern-.15ex I}}  

\@ifundefined{showcaptionsetup}{}{%
 \PassOptionsToPackage{caption=false}{subfig}}
\usepackage{subfig}
\makeatother

\providecommand{\definitionname}{Definition}
\providecommand{\lemmaname}{Lemma}
\providecommand{\propositionname}{Proposition}
\providecommand{\remarkname}{Remark}
\providecommand{\theoremname}{Theorem}

\begin{document}
\title{\textbf{Quantification of Risk in Classical Models of Finance}}
\author{Alois Pichler\thanks{Funded by Deutsche Forschungsgemeinschaft (DFG, German Research Foundation)
\textendash{} Project-ID~416228727~\textendash{} SFB~1410} \and  Ruben Schlotter\thanks{Both authors: Technische Universität Chemnitz, 09126 Chemnitz, Germany.
Contact: \protect\href{mailto:ruben.schlotter@math.tu-chemnitz.de}{ruben.schlotter@math.tu-chemnitz.de}}}
\maketitle
\begin{abstract}
This paper enhances the pricing of derivatives as well as optimal
control problems to a level comprising risk. We employ nested risk
measures to quantify risk, investigate the limiting behavior of nested
risk measures within the classical models in finance and characterize
existence of the risk-averse limit. As a result we demonstrate that
the nested limit is unique, irrespective of the initially chosen risk
measure. Within the classical models risk aversion gives rise to a
stream of risk premiums, comparable to dividend payments. In this
context we connect coherent risk measures with the Sharpe ratio from
modern portfolio theory and extract the Z-spread\textemdash a widely
accepted quantity in economics to hedge risk.

The results for European option pricing are then extended to risk-averse
American options, where we study the impact of risk on the price as
well as the optimal time to exercise the option.

We also extend Merton's optimal consumption problem to the risk-averse
setting.

\bigskip{}


\textbf{Classification:} 90C15, 60B05, 62P05
\end{abstract}

\section{\label{sec:Introduction}Introduction}

This paper studies discrete classical models in finance under risk
aversion and their behavior in a high-frequency setting. Using nested
risk measures we first study risk aversion in the multiperiod model.

We develop risk aversion in a discrete time and discrete space setting
and find an important consistency property of nested risk measures.
This consistency property, termed \emph{divisibility}, is crucial
in high-frequency trading environments. For this, our study of risk-averse
models extends to continuous time processes as well. This very property
allows consistent decision making, i.e., decisions, which are independent
of individually chosen discretizations or trading frequencies. Our
results also give rise to a generalized Black\textendash Scholes framework,
which incorporates risk aversion in addition.

\citet{Riedel2004-R} has introduced risk measures in a dynamic setting.
Later, \citet{Cheridito2004} study risk measures for bounded càdlàg
processes and \citet{cher2006-R} also discuss risk measures in a
discrete time setting. \citet{Ruszczynski-R} introduce nested risk
measures, for which \citet{Philpott2013} provide an economic interpretation
as an insurance premium on a rolling horizon basis. For a recent discussion
on risk measures and dynamic optimization we refer to \citet{DeLara2015-R}.
Applications can be found in \citet{PhilpottMatos} or \citet{Maggioni2012},
e.g., where stochastic dual dynamic programming methods are addressed,
see also \citet{Guigues2012}.

Divisibility is an indispensable prerequisite in defining an infinitesimal
generator based on discretizations. This generator, called \emph{risk
generator}, constitutes the risk-averse assessment of the dynamics
of the underlying stochastic process. Using the risk generator we
characterize the existence of the risk-averse limit of discrete pricing
models. For coherent risk measures and Itô diffusion processes the
risk generator constitutes a nonlinear operator, comparable to the
classical infinitesimal generator but with an additional term, accounting
for risk, which takes the form 
\[
s_{\rho}\,\left|\sigma\,\partial_{x}\,(\cdot)\right|.
\]
Here, $s_{\rho}$ is a scalar expressing the degree of risk aversion
and $\sigma$ is the volatility of the diffusion process describing
the asset price. It turns out that the risk generator only depends
on the risk measure through the coefficient of risk aversion $s_{\rho}$.
This surprising feature has important conceptual implications, as
evaluating a risk measure is often an optimization problem itself.
As well we derive that the scaling quantity $s_{\rho}$ allows the
economic interpretation of a Sharpe ratio and $s_{\rho}\cdot\sigma$
is the Z-spread.

Using the risk generator we derive a nonlinear Black\textendash Scholes
equation, which we relate to the Black\textendash Scholes formula
for dividend paying stocks proposed by \citet{Merton1973}. Moreover
we relate risk-averse pricing models to foreign exchange options models
as in \citet{Garman1983}. Nonlinear Black\textendash Scholes equations
have been discussed previously in \citet{Barles1998} and \citet{Sevcovic2016}
in the context of modeling transaction costs. There, the nonlinearity
is in the second derivative. In contrast, risk aversion leads to drift
uncertainty and causes nonlinearity in the first derivative.

Very different to our approach, \citet{Stadje2010} studies the convergence
properties of discretizations of dynamic risk measures based on backwards
stochastic differential equations introduced in \citet{Pardoux1990}
(see also \citet{Delong2013} for an overview). \citet{RuszczynskiHJB-R}
then derive risk-averse Hamilton\textendash Jacobi\textendash Bellmann
equations based on these backwards stochastic differential equations. 

For coherent risk measures we derive an explicit solution for the
European option pricing problem. We show that risk aversion expressed
via coherent risk measures can be interpreted either as an extra dividend
payment or capital injection. Furthermore we relate risk-aversion
to a change of currency as in the foreign exchange option model. The
amount of the dividend payment or, equivalently, the interest rate
in the risk-averse currency, is given by a multiple of the Sharpe
ratio and the volatility of the underlying stock. This ratio, which
expresses risk aversion, arises for any coherent risk measure and
does not depend on a specific market model such as the Black\textendash Scholes
model. However, as our focus is on classical models, we restrict ourselves
to Itô diffusion processes.

Using a free boundary formulation we extend the analysis from European
to American option pricing. For the Black\textendash Scholes option
pricing of European and American options, risk-aversion naturally
leads to a bid-ask spread, which we quantify explicitly.

Similarly we extend the Merton optimal consumption problem to a risk-averse
setting. We elaborate on the optimal controls and show that risk-aversion
reduces the investment in risky assets and increases consumption.
We observe the same pattern as for European and American options,
that is, risk-aversion corrects the drift of the underlying market
model. For all classical models discussed here, the risk-averse assessment
still allows explicit pricing and control formulae. 

\section{Preliminaries on risk measures}

Recall the definition of \emph{law invariant, coherent risk measures}
$\rho\colon L\to\mathbb{R}$ defined on some vector space $L$ of
$\mathbb{R}$\nobreakdash-valued random variables first. They satisfy
the following axioms introduced by \citet{Artzner1999-R}.
\begin{enumerate}[label=A\arabic*., ref=A\arabic*, noitemsep]
\item \label{enu:Monotonicity}Monotonicity: $\rho(Y)\le\rho(Y^{\prime})$,
provided that $Y\le Y^{\prime}$ almost surely;
\item \label{enu:equivariance}Translation equivariance: $\rho(Y+c)=\rho(Y)+c$
for $c\in\mathbb{R}$;
\item \label{enu:Convexity} Subadditivity: $\rho\big(Y+Y^{\prime}\big)\le\rho(Y)+\rho(Y^{\prime})$;
\item \label{enu:Homogeneous}Positive homogeneity: $\rho(\lambda\,Y)=\lambda\,\rho(Y)$
for $\lambda\ge0$;
\end{enumerate}
\begin{enumerate}[resume, label=A\arabic*., ref=A\arabic*, noitemsep]
\item \label{enu:LawInvariance}Law invariance: $\rho(Y)=\rho(Y^{\prime})$,
whenever $Y$ and $Y^{\prime}$ have the same law, i.e., $P(Y\le y)=P(Y^{\prime}\le y)$
for all $y\in\mathbb{R}$.
\end{enumerate}
The expectation ($\rho(Y)=\E Y$) is also a law invariant coherent
risk measure, expressing risk-neutral behavior. In contrast to the
risk-neutral setting, the risk-averse setting distinguishes between
$\rho(Y)$ and $-\rho(-Y)$. As a result of\footnote{The inequality $0=\rho(Y-Y)\leq\rho(Y)+\rho(-Y)$ implies that $-\rho(-Y)\leq\rho(Y)$.}
\[
-\rho(-Y)\le\rho(Y)
\]
we will later identify $\rho(Y)$ with the seller's ask price and
$-\rho(-Y)$ with the buyer's bid price in the option pricing problems
discussed below.

\subsection{Nested risk measures}

We consider a filtered probability space $\left(\Omega,\mathcal{F},(\mathcal{F}_{t})_{t\in\mathcal{T}},P\right)$
and associate $t\in\mathcal{T}$ with \emph{stage} or\emph{ time}.
For the discussion of risk in a dynamic setting we introduce nested
risk measures corresponding to the evolution of risk over time. Nested
risk measures are compositions of conditional risk measures (cf.\ \citet{PflugRomisch2007-R}).

Recall that a coherent risk measures $\rho\colon L\to\mathbb{R}$
can be represented by 
\begin{equation}
\rho(Y)=\sup_{Q\in\mathcal{Q}}\,\E_{Q}Y,\label{eq:dualRep}
\end{equation}
where $\mathcal{Q}$ is a convex set of probability measures absolutely
continuous with respect to $P$ (cf.\ also \citet{Delbaen02coherentrisk-R}).
We assume throughout that $\rho\colon L^{p}\to\mathbb{R}$ for some
fixed $p\geq1$. Following \citet{Ruszczynski-R}, we then introduce
conditional versions $\rho^{t}\colon L^{p}\to L^{p}(\Omega,\mathcal{F}_{t},P)$
of the risk measure $\rho$ conditioned on the sigma algebra $\mathcal{F}_{t}$.
Note that the conditional risk measures $\rho^{t}$ satisfy conditional
versions of the Axioms~\ref{enu:Monotonicity}\textendash \ref{enu:LawInvariance}
above. For the construction of $\rho^{t}$ and further details we
refer the interested reader also to \citet[Section 6.8.2]{Shapiro2014}.

We now introduce nested risk measures in discrete time.
\begin{defn}[Nested risk measures]
\label{def:nRM} The \emph{nested risk measure} for the partition
$\mathcal{P}=\left(t_{0},t_{1},\dots,t_{n}\right)$ at times $t_{0}<\ldots<t_{n}$
is 
\begin{equation}
\rho^{\mathcal{P}}(Y)\coloneqq\rho^{t_{0}}\left(\rho^{t_{1}}\left(\ldots\rho^{t_{n}}(Y)\ldots\right)\right),\label{eq:17}
\end{equation}
where $(\rho^{t_{i}})_{i=0}^{n}$ is a family of conditional risk
measures.
\end{defn}

Similar as above, we distinguish the buyer's and seller's perspective
and consider the bid price
\[
-\rho^{\mathcal{P}}(-Y)=-\rho^{t_{0}}\left(\rho^{t_{1}}\left(\ldots\rho^{t_{n}}(-Y)\ldots\right)\right),
\]
as well as the ask price in~\eqref{eq:17}.

\subsection{Nested risk measures for discrete processes}

To elaborate key properties of nested risk measures as defined in~\eqref{eq:17}
we discuss the binomial model, well-known from finance, by employing
the mean semi-deviation, a coherent risk measure satisfying all Axioms~\ref{enu:Monotonicity}\textendash \ref{enu:LawInvariance}
above. Particularly, we expose that only specific choices of parameters
can lead to consistent models. 
\begin{defn}[Semi-deviation]
\label{def:mSD}The mean semi-deviation risk measure of order $p\geq1$
and $Y\in L^{p}$ at level $\beta\in[0,1]$ is 
\[
\SD_{p,\beta}(Y):=\E Y+\beta\left\Vert \left(Y-\E Y\right)_{+}\right\Vert _{p}.
\]
\end{defn}

\paragraph{The binomial model.}

Consider the stochastic process $S=(S_{0},\dots,S_{T})$ with initial
state $S_{0}$ and Markovian transitions with 
\begin{equation}
P\left(S_{t+\Delta t}=S_{t}\cdot e^{\pm\sigma\sqrt{\Delta t}}\right)=p_{\pm},\label{eq:Binomial}
\end{equation}
where 
\[
p_{+}:=p:=\frac{e^{r\Delta t}-e^{-\sigma\sqrt{\Delta t}}}{e^{\sigma\sqrt{\Delta t}}-e^{-\sigma\sqrt{\Delta t}}}\text{ and }p_{-}:=1-p_{+}.
\]
It holds that $\E S_{t+\Delta t}=pS_{t}e^{\sigma\sqrt{\Delta t}}+(1-p)S_{t}e^{-\sigma\sqrt{\Delta t}}=S_{t}e^{r\Delta t}$.
In stochastic finance, the process~$S$ models the evolution of a
stock over time with respect to the risk-neutral risk measure, where
$r$ is the risk free interest rate. 

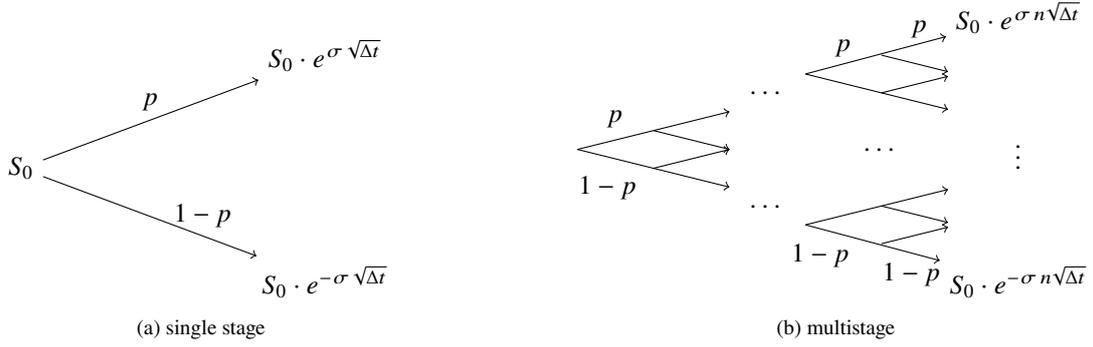
\begin{figure}[H]
\centering{}\subfloat[\label{fig:single}single stage]{\begin{centering}
\begin{tikzpicture}[yscale=0.5]	
\node (o0)  at (0,3)  {$S_0$};
\node (o1)  at (4, 6) {$S_{0}\cdot e^{\sigma\,\sqrt{\Delta t}}$};
\node (o2)  at (4, 0) {$S_{0}\cdot e^{-\sigma\,\sqrt{\Delta t}}$};
\draw[->, thin] (o0) -- node[above]{$p$}(o1);
\draw[->, thin] (o0) -- node[above, near end]{$1-p$}(o2);
\end{tikzpicture}
\par\end{centering}
}\hfill{}\subfloat[\label{fig:multi}multistage]{\begin{centering}
\begin{tikzpicture}[yscale=0.5]	
\node (on)  at (5.8,6.5)  {$S_0 \cdot e^{\sigma\,n \sqrt{\Delta t}}$};
\node (on1)  at (5,5) {};				
\node (on2)  at (5,4) {};				
\node (o2)  at (5,2) {};				 
\node (o1)  at (5,1) {};				
\node (o0)  at (5.8,-0.5) {$S_0 \cdot e^{-\sigma\,n \sqrt{\Delta t}}$};				

\draw[->, thin] (0,3.0)--node[above]{}(2,4.0);
\node at (0.5,3.8) {$p$};
\draw[->, thin] (1,3.5)--(2,3.0);
\draw[->, thin] (0,3.0)--node[below]{}(2,2);
\node at (0.4,2.0) {$1-p$};
\draw[->, thin] (1,2.5)--(2,3.0);

\node       at (2.5,4.5) {$\dots$};
\draw[->, thin] (3,5)--node[above]{}(on);
\node at (3.5,5.7) {$p$};
\node at (4.5,6.2) {$p$};
\draw[->, thin] (4,5.5)--(on1);
\draw[->, thin] (3,5)--(on2);
\draw[->, thin] (4,4.5)--(on1);

\node       at (4,3) {$\dots$};
\node       at (5.8,3) {$\vdots$};

\node       at (2.5,1.5) {$\dots$};
\draw[->, thin] (3,1)--(o2);
\draw[->, thin] (4,1.5)--(o1);
\draw[->, thin] (3,1)--node[below]{}(o0);
\node at (3.2,0.2) {$1-p$};
\node at (4.4,-0.3) {$1-p$};
\draw[->, thin] (4,0.5)--(o1);
\end{tikzpicture}
\par\end{centering}
}\caption{\label{fig:prices-1}Binomial option pricing model}
\end{figure}

We can evaluate various classical coherent risk measures for this
binomial model explicitly. The following remark addresses the mean
semi-deviation for the one-period binomial model (cf.\ Figure~\ref{fig:single})
as well as the nested mean semi-deviation for the $n$-period model
in (Figure~\ref{fig:multi}).
\begin{rem}[The mean semi-deviation for the binomial model]
\label{rem:risk-binom}Consider the single stage setting in Figure~\ref{fig:single}
first. The risk-averse bid price for the stock $S_{\Delta t}$ employing
the mean semi-deviation $\SD_{1,\beta}$ of order $1$ with risk level
$\beta$ in the binomial model is
\begin{align*}
-\SD_{1,\beta}(-S_{\Delta t}) & =\E S_{\Delta t}-\beta\E\left(-S_{\Delta t}+\E S_{\Delta t}\right)_{+}\\
 & =pS_{0}e^{\sigma\sqrt{\Delta t}}+(1-p)S_{0}e^{-\sigma\sqrt{\Delta t}}-\beta\,p(1-p)\left(S_{0}e^{\sigma\sqrt{\Delta t}}-S_{0}e^{-\sigma\sqrt{\Delta t}}\right).
\end{align*}
Involving the new probability weights
\begin{align}
\widetilde{p} & :=p\big(1-\beta(1-p)\big)\label{eq:15}
\end{align}
we find
\[
-\SD_{1,\beta}(-S_{\Delta t})=\widetilde{\E}S_{\Delta t}.
\]

We now repeat this observation in $n$ stages and consider an $n$-period
binomial model with step size $\Delta t\coloneqq\frac{T}{n}$, i.e.,
$\mathcal{P}=\left(0,\Delta t,2\Delta t,\dots,T\right)$, cf.\ Figure~\ref{fig:multi}.
The nested mean semi-deviation for the vector of constant risk levels
$\beta=(\widetilde{\beta},\dots,\widetilde{\beta})>0$ satisfies
\[
-\SD_{1,\beta}^{\mathcal{P}}(-S_{T})=-\SD_{1,\widetilde{\beta}}\left(\dots\SD_{1,\widetilde{\beta}}\left(-S_{T}\right)\dots\right)=\widetilde{\E}S_{T},
\]
where the last expectation is with respect to the probability measure
\[
\widetilde{P}\left(S_{T}=S_{0}e^{\sigma\left(2k\sqrt{\Delta t}-n\sqrt{\Delta t}\right)}\right)=\binom{n}{k}\widetilde{p}^{k}(1-\widetilde{p})^{n-k},\qquad k=0,\dots,n.
\]
The limit 
\begin{equation}
\frac{1}{\sqrt{n\,\widetilde{p}(1-\widetilde{p})}}\left(\frac{\frac{1}{\sigma}\log\frac{S_{T}}{S_{0}}+n\sqrt{\Delta t}}{2\sqrt{\Delta t}}-n\,\widetilde{p}\right)\label{eq:18}
\end{equation}
is non-degenerate for $n\to\infty$, provided that $\widetilde{p}\to\frac{1}{2}$.
Based on the central limit theorem, the limit~\eqref{eq:18} then
follows a standard normal distribution.

Hence, specific choices of the parameter $\beta$ in~\eqref{eq:15}
depending on the discretization have to be considered. To this end
we introduce the notion of divisible families of risk measures below
and return to this example in Section~\ref{subsec:consistency}.
\end{rem}

\section{The risk-averse limit of discrete option pricing models}

Most well-known coherent risk measures in the literature as the Average
Value-at-Risk, the Entropic Value-at-Risk as well as the mean semi-deviation
involve a parameter which accounts for the degree of risk aversion.
As Remark~\ref{rem:risk-binom} elaborates, the nested risk-averse
binomial model does not necessarily lead to a well-defined limit.
It is essential to relate the coefficient of risk aversion of the
conditional risk measures to its time period. We therefore introduce
the notion of \emph{divisible} coherent risk measures. The divisibility
property is central in discussing the limiting behavior of risk-averse
economic models.
\begin{defn}[Divisible families of risk measures]
\label{def:divisibility}Let $p\geq1$ be fixed. A family $\rho=\left\{ \left.\rho_{\Delta t}\colon L^{p}\to\mathbb{R}\right|\Delta t>0\right\} $
of coherent measures of risk is called \emph{divisible}, if the following
two conditions are satisfied:
\begin{enumerate}
\item For $W\sim\mathcal{N}(0,1)$ normally distributed, 
\begin{equation}
\lim_{\Delta t\downarrow0}\,\frac{\rho_{\Delta t}(\sqrt{\Delta t}\cdot W)}{\Delta t}=s_{\rho}\label{eq:17-1}
\end{equation}
for some $s_{\rho}\ge0$.
\item Moreover there is a constant $C>0$ (independent of $Y$ and $\Delta t$)
such that 
\[
\rho_{\Delta t}(Y)\leq C\sqrt{\Delta t}\left\Vert Y\right\Vert _{p}
\]
for all $Y\in L^{p}$ with $\E Y=0$.
\end{enumerate}
We call a nested risk measure $\rho^{\mathcal{P}}$ divisible if every
conditional risk measure is divisible, i.e., the limit in~\eqref{eq:17-1}
holds for random variables which are conditionally normally distributed
and 
\[
\rho_{\Delta t}^{t}(Y)\leq C\sqrt{\Delta t}\E\big(\left|Y\right|^{p}\mid\mathcal{F}_{t}\big)^{\frac{1}{p}}
\]
for some constant $C>0$.
\end{defn}

\begin{rem}
The (conditional) expectation is divisible with $s_{\mathbb{E}}=0$.
For many other risk measures, the parameters can be adjusted. Candidates
for risk measures satisfying this condition are spectral risk measures
for which the spectral density is bounded in the $L^{q}$ norm for
$q=\frac{p}{p-1}$. The mean semi-deviation risk measure satisfies
the divisibility property as well.
\end{rem}

\begin{lem}
\label{lem:gauss}For $p\ge1$ and $\beta\ge0$, the family 
\[
\left\{ \SD_{p,\beta;\Delta t}:=\SD_{p,\beta\cdot\sqrt{\Delta t}}\right\} ,\quad\Delta t>0,
\]
of mean semi-deviations is divisible with limit 
\[
s_{\SD_{p,\beta}}=\beta\left(2\pi\right)^{-\frac{1}{2p}}2^{\frac{1}{2}-\frac{1}{2p}}\cdot\Gamma\left(\frac{p+1}{2}\right)^{\frac{1}{p}}.
\]
\end{lem}

\begin{proof}
The second part of Definition~\ref{def:divisibility} is satisfied
as for $Y\in L^{p}$ such that $\E Y=0$ we have
\[
\SD_{\beta\sqrt{\Delta t},p}(Y)=\beta\sqrt{\Delta t}\left\Vert Y_{+}\right\Vert _{p}\leq\beta\sqrt{\Delta t}\left\Vert Y\right\Vert _{p}.
\]
Let $W\sim\mathcal{N}(0,1)$, then 
\begin{align*}
\E\left(\sqrt{\Delta t}W_{+}\right)^{p} & =\int_{\mathbb{R}}\max(w,0)^{p}\cdot\frac{1}{\sqrt{2\pi\Delta t}}e^{-\frac{w^{2}}{2\Delta t}}\,\mathrm{d}w=\frac{1}{\sqrt{2\pi\Delta t}}\int_{0}^{\infty}w^{p}\cdot e^{-\frac{w^{2}}{2\Delta t}}\,\mathrm{d}w.
\end{align*}
Employing the Gamma function, the latter integral is
\begin{align*}
\frac{1}{\sqrt{2\pi\Delta t}}\int_{0}^{\infty}w^{p}\cdot e^{-\frac{w^{2}}{2\Delta t}}\,\mathrm{d}w & =\frac{1}{\sqrt{2\pi}}2^{\frac{p-1}{2}}\Gamma\left(\frac{p+1}{2}\right)\Delta t^{\frac{p}{2}}.
\end{align*}
Taking the $p$-th root and multiplying by $\beta\sqrt{\Delta t}$
we obtain
\[
\frac{\SD_{p,\beta\sqrt{\Delta t}}(\sqrt{\Delta t}W)}{\Delta t}=\beta\left(2\pi\right)^{-\frac{1}{2p}}2^{\frac{1}{2}-\frac{1}{2p}}\cdot\Gamma\left(\frac{p+1}{2}\right)^{\frac{1}{p}},
\]
the assertion. 
\end{proof}
We now extend nested risk measures to continuous time and demonstrate
that the extension is well-defined for divisible families of risk
measures. As a result, we show that the risk-averse binomial option
pricing model converges exactly for divisible families of risk measures.
\begin{defn}[Nested risk measures]
\label{def:nRM2}Let $T>0$, $t\in[0,T)$ and let $\rho^{\mathcal{P}}$
be divisible for every partition $\mathcal{P}\subset[t,T]$, cf.\ Definition~\ref{def:nRM}.
The \emph{nested risk measure $\rho^{t:T}$ in continuous time for
a random variable $Y$} is
\begin{equation}
\rho^{t:T}\left(Y\left|\,\mathcal{F}_{t}\right.\right):=\lim_{\mathcal{P}\subset[t,T]}\,\rho^{\mathcal{P}}\left(Y\left|\,\mathcal{F}_{t}\right.\right)\qquad\text{almost surely},\label{eq:CnRM}
\end{equation}
where the almost sure limit is among all partitions $\mathcal{P}\subset[t,T]$
with mesh size $\left\Vert \mathcal{P}\right\Vert \coloneqq\max_{i=1,\dots,n}t_{i}-t_{i-1}$
tending to zero for those random variables $Y$, for which the limit
exists. 
\end{defn}

The following proposition evaluates the nested mean-semideviation
for the Wiener process, the basic building block of diffusion processes
and thus illustrates the main purpose of the divisibility condition.
\begin{prop}[Nested mean semi-deviation for the Wiener process]
\label{prop:randomWalk}Let $W=(W_{t})_{t\in\mathcal{P}}$ be a Wiener
process and $\mathcal{P}=\left(t_{0},t_{1},\dots,t_{n}\right)$ a
partition of $[0,T]$ with $\Delta t_{i}:=t_{i+1}-t_{i}$. For the
family of conditional risk measures $\left(\SD_{p,\beta_{t_{i}}\cdot\sqrt{\Delta t_{i}}}(\cdot\mid\mathcal{F}_{t_{i}})\right)_{t_{i}\in\mathcal{P}}$,
the nested mean semi-deviation is
\begin{equation}
\SD_{p,\beta}^{\mathcal{P}}(W_{T})=\sum_{i=0}^{n-1}\beta_{t_{i}}\Delta t_{i}\cdot\left(2\pi\right)^{-\frac{1}{2p}}2^{-\frac{1}{2}}\Gamma\left(\frac{p+1}{2}\right)^{\frac{1}{p}},\label{eq:D-W}
\end{equation}
where $\beta=(\beta_{t_{0}},\dots,\beta_{t_{n}})$ is a vector of
risk levels. 
\end{prop}

\begin{proof}
Note that $W_{t_{i+1}}-W_{t_{i}}\sim\mathcal{N}(0,t_{i+1}-t_{i})$
and the conditional mean semi-deviation is (using conditional translation
equivariance~\ref{enu:equivariance})
\begin{align*}
\SD_{p,\beta_{t_{i}}\cdot\sqrt{\Delta t_{i}}}\left(W_{t_{i+1}}\left|\,W_{t_{i}}\right.\right) & =W_{t_{i}}+\SD_{p,\beta_{t_{i}};\sqrt{\Delta t_{i}}}\left(W_{t_{i+1}}-W_{t_{i}}\left|\,W_{t_{i}}\right.\right).
\end{align*}
As Brownian motion has independent and stationary increments with
mean zero the calculation in the proof of Lemma~\ref{lem:gauss}
shows that
\begin{align*}
\SD_{p,\beta_{t_{i}}\cdot\sqrt{\Delta t_{i}}}\left(W_{t_{i+1}}\left|\,W_{t_{i}}\right.\right) & =W_{t_{i}}+\beta_{t_{i}}\Delta t_{i}\cdot\left(2\pi\right)^{-\frac{1}{2p}}2^{-\frac{1}{2}}\Gamma\left(\frac{p+1}{2}\right)^{\frac{1}{p}}.
\end{align*}
Iterating as in Definition~\ref{def:nRM} shows
\begin{align*}
\SD_{p,\beta}^{\mathcal{P}}(W_{T}) & =\sum_{i=0}^{n-1}\beta_{t_{i}}\Delta t_{i}\cdot\left(2\pi\right)^{-\frac{1}{2p}}2^{-\frac{1}{2}}\Gamma\left(\frac{p+1}{2}\right)^{\frac{1}{p}},
\end{align*}
the assertion.
\end{proof}
\begin{rem}
For constant risk levels $\beta_{t_{i}}=\widetilde{\beta}$ we obtain
\[
\SD_{p,\beta}^{\mathcal{P}}(W_{T})=\sum_{i=0}^{n-1}\Delta t_{i}\cdot\widetilde{\beta}\cdot\left(2\pi\right)^{-\frac{1}{2p}}2^{-\frac{1}{2}}\Gamma\left(\frac{p+1}{2}\right)^{\frac{1}{p}}=T\cdot s_{\SD_{p,\widetilde{\beta}}},
\]
the accumulated risk grows linearly in time.
\end{rem}

\subsection{The risk generator}

This section addresses nested risk measures for Itô processes. Furthermore,
we characterize convergence under risk using a natural condition involving
normal random variables and introduce a nonlinear operator, the \emph{risk
generator,} which also allows discussing risk-averse optimal control
problems. 

It is well-known that the binomial model in Figure~\ref{fig:multi}
converges to the geometric Brownian motion. We therefore discuss Itô
processes $(X_{s})_{s\in\mathcal{T}}$ solving the stochastic differential
equation
\begin{align}
\mathrm{d}X_{s} & =b(s,X_{s})\,\mathrm{d}s+\sigma(s,X_{s})\,\mathrm{d}W_{s},\quad s\in\mathcal{T},\label{eq:SDE}\\
X_{t} & =x\nonumber 
\end{align}
for $\mathcal{T}=[t,\,T]$. We assume that $X$ following~\eqref{eq:SDE}
is well-defined and satisfy the so-called \emph{usual conditions}
of \citet[Theorem~5.2.1]{Oksendal2003}.

We introduce the \emph{risk generator} for divisible families of coherent
risk measures. The risk generator describes the momentary evolution
of the risk of the stochastic process.
\begin{defn}[Risk generator]
\label{def:RiskGen}Let $X=(X_{t})_{t}$ be a continuous time process
and $(\rho_{\Delta t})_{\Delta t}$ be a family of divisible risk
measures. The \emph{risk generator} \emph{based on }$(\rho_{\Delta t})_{\Delta t}$
is 
\begin{equation}
\mathcal{R}_{\rho}\Phi(t,x):=\lim_{\Delta t\downarrow0}\frac{1}{\Delta t}\Bigl(\rho_{\Delta t}\bigl(\Phi(t+\Delta t,X_{t+\Delta t})\left|\,X_{t}=x\right.\bigr)-\Phi(t,x)\Bigr)\label{eq:riskGen}
\end{equation}
for those functions $\Phi\colon\mathcal{T}\times\mathbb{R}\to\mathbb{R}$,
for which the limit exists.
\end{defn}

Using the ideas from Proposition~\ref{prop:randomWalk} we obtain
explicit expressions for the risk generator for Itô diffusion processes.
\begin{prop}[Risk generator]
\label{prop:RGen}Let the family $(\rho_{\Delta t})_{\Delta t}$
be divisible for some $p\ge1$ fixed. Let $X$ be the solution of~\eqref{eq:SDE}
and $\Phi\in C^{2}(\mathcal{T}\times\mathbb{R})$ such that $\sigma\,\Phi_{x}$
is Hölder continuous for $\alpha>0$ in $p$\nobreakdash-th mean,
i.e., there exists $C>0$ such that $\E C^{p}<\infty$ and
\begin{equation}
\left|\left(\sigma\,\Phi_{x}\right)(t,X_{t})-\left(\sigma\,\Phi_{x}\right)(s,X_{s})\right|\leq C\cdot\left|t-s\right|^{\alpha},\qquad s,t\in\mathcal{T}.\label{eq:hoelder-continuity}
\end{equation}
Then the risk generator based on $(\rho_{\Delta t})_{\Delta t}$ is
given by the nonlinear differential operator
\begin{align}
\mathcal{R}_{\rho}\Phi(t,x) & =\left(\Phi_{t}+b\,\Phi_{x}+\frac{\sigma^{2}}{2}\Phi_{xx}+s_{\rho}\cdot\left|\sigma\,\Phi_{x}\right|\right)(t,x).\label{eq:21}
\end{align}
\end{prop}

\begin{rem}
In the appendix we provide a sufficient condition for the Assumption~\eqref{eq:hoelder-continuity}.
 
\end{rem}

\begin{proof}
By assumption, $\Phi\in C^{2}(\mathcal{T}\times\mathbb{R})$ and hence
we may apply Itô's formula. For convenience and ease of notation we
set $f_{1}(t,x):=\left(\Phi_{t}+b\,\Phi_{x}+\frac{\sigma^{2}}{2}\Phi_{xx}\right)(t,x)$
and $f_{2}(t,x):=\left(\sigma\,\Phi_{x}\right)(t,x)$. In this setting,
Eq.~\eqref{eq:riskGen} rewrites as
\begin{align*}
\mathcal{R}_{\rho}\Phi(t,x) & =\lim_{\Delta t\downarrow0}\,\frac{1}{\Delta t}\rho_{\Delta t}^{t}\left[\left.\int_{t}^{t+\Delta t}f_{1}(s,X_{s})\,\mathrm{d}s+\int_{t}^{t+\Delta t}f_{2}(s,X_{s})\,\mathrm{d}W_{s}\right|\,X_{t}=x\right].
\end{align*}
To show~\eqref{eq:21} for each fixed $(t,x)$ it is enough to show
that
\begin{equation}
\left|\mathcal{R}_{\rho}\Phi(t,x)-f_{1}(t,x)-s_{\rho}\left|f_{2}(t,x)\right|\right|\leq0.\label{eq:7-1}
\end{equation}
Using the properties \ref{enu:equivariance}\textendash \ref{enu:Homogeneous}
of coherent risk measures together with the triangle inequality we
bound the left side of~\eqref{eq:7-1} by
\begin{align}
\lim_{\Delta t\downarrow0} & \left|\rho_{\Delta t}^{t}\left[\left.\frac{1}{\Delta t}\int_{t}^{t+\Delta t}f_{1}(s,X_{s})\mathrm{d}s-f_{1}(t,x)\right|\,X_{t}=x\right]\right|\nonumber \\
 & +\lim_{\Delta t\downarrow0}\left|\rho_{\Delta t}^{t}\left[\left.\frac{1}{\Delta t}\int_{t}^{t+\Delta t}f_{2}(s,X_{s})\mathrm{d}W_{s}-s_{\rho}\left|f_{2}(t,x)\right|\right|\,X_{t}=x\right]\right|.\label{eq:??}
\end{align}
We continue by looking at each term separately. Note that $s\mapsto f_{1}(s,X_{s})-f_{1}(t,x)$
is continuous almost surely and hence the mean value theorem for definite
integrals implies that there exists a $\xi\in[t,t+\Delta t]$ such
that
\[
\frac{1}{\Delta t}\int_{t}^{t+\Delta t}f_{1}(s,X_{s})\mathrm{d}s-f_{1}(t,x)=f_{1}(\xi,X_{\xi})-f_{1}(t,x),\quad\text{almost surely}.
\]
From continuity of $\rho$ in the $L^{p}$ norm we may conclude
\[
\lim_{\Delta t\downarrow0}\,\frac{1}{\Delta t}\rho_{\Delta t}^{t}\left(\left.\left|\int_{t}^{t+\Delta t}f_{1}(s,X_{s})-f_{1}(t,x)\,\mathrm{ds}\right|\,\right|X_{t}=x\right)=0.
\]

Note that the stochastic integral term in~\eqref{eq:??} can be bounded
by

\begin{align*}
\rho_{\Delta t}^{t}\left[\left.\frac{1}{\Delta t}\int_{t}^{t+\Delta t}f_{2}(s,X_{s})\mathrm{d}W_{s}\right|\,X_{t}=x\right] & \leq\rho_{\Delta t}^{t}\left[\left.\frac{1}{\Delta t}\int_{t}^{t+\Delta t}f_{2}(s,X_{s})-f_{2}(t,x)\mathrm{d}W_{s}\right|\,X_{t}=x\right]\\
 & +\,\rho_{\Delta t}^{t}\left[\left.\frac{1}{\Delta t}\int_{t}^{t+\Delta t}f_{2}(t,x)\mathrm{d}W_{s}\right|\,X_{t}=x\right],
\end{align*}
where $\rho_{\Delta t}^{t}\left[\left.\frac{1}{\Delta t}\int_{t}^{t+\Delta t}f_{2}(t,x)\mathrm{d}W_{s}\right|\,X_{t}=x\right]$
converges to $s_{\rho}\left|f_{2}(t,x)\right|$ and hence
\begin{align*}
\eqref{eq:??} & \leq\lim_{\Delta t\downarrow0}\left|\rho_{\Delta t}^{t}\left[\left.\frac{1}{\Delta t}\int_{t}^{t+\Delta t}f_{2}(s,X_{s})-f_{2}(t,x)\mathrm{d}W_{s}\right|\,X_{t}=x\right]\right|.
\end{align*}
Furthermore, the stochastic integral $M_{\Delta t}:=\int_{t}^{t+\Delta t}f_{2}(s,X_{s})-f_{2}(t,x)\mathrm{d}W_{s}$
is a continuous martingale with $M_{0}=0$ and by divisibility there
exists a constant $\widetilde{C}$ independent of $\Delta t$ and
$M_{\Delta t}$ such that
\[
\rho_{\Delta t}^{t}(M_{\Delta t})\leq\widetilde{C}\sqrt{\Delta t}\cdot\left\Vert M_{\Delta t}\right\Vert _{p}.
\]
Applying the Burkholder\textendash Davis\textendash Gundy inequality
implies the upper bound
\begin{align*}
\left\Vert M_{\Delta t}\right\Vert _{p} & \leq C_{\mathit{BDG}}\left[\E\left|\int_{t}^{t+\Delta t}\left(f_{2}(s,X_{s})-f_{2}(t,x)\right)^{2}\mathrm{d}s\right|^{\frac{p}{2}}\right]^{\frac{1}{p}}
\end{align*}
for some constant $C_{\mathit{BDG}}$ depending on $p$. By assumption
there exists a random $C>0$ such that 
\begin{align*}
\E\left(\int_{t}^{t+\Delta t}\left(f_{2}(s,X_{s})-f_{2}(t,x)\right)^{2}\mathrm{d}s\right)^{\frac{p}{2}} & \leq\E\left(\int_{t}^{t+\Delta t}C^{2}\left|s-t\right|^{2\alpha}\mathrm{d}s\right)^{\frac{p}{2}}=\left(\frac{\Delta t^{2\alpha+1}}{2\alpha+1}\right)^{\frac{p}{2}}\E C^{p}.
\end{align*}
Therefore,
\begin{align*}
\rho_{\Delta t}^{t}(M_{\Delta t}) & \leq\widetilde{C}\cdot C_{\mathit{BDG}}\sqrt{\Delta t}\left\Vert C\right\Vert _{p}\left(\frac{\Delta t^{2\alpha+1}}{2\alpha+1}\right)^{\frac{1}{2}}=\frac{\widetilde{C}\cdot C_{\mathit{BDG}}}{\sqrt{2\alpha+1}}\left\Vert C\right\Vert _{p}\Delta t^{1+\alpha},
\end{align*}
such that $\frac{1}{\Delta t}\rho_{\Delta t}^{t}(M_{\Delta t})$ vanishes
for $\Delta t\to0$, which concludes the proof.
\end{proof}
\begin{rem}[Relation to $g$-expectation]
The risk generator $\mathcal{R}_{\rho}$ can be decomposed as the
sum of the classical generator \emph{plus} the nonlinear term $s_{\rho}\left|\sigma\frac{\partial\Phi}{\partial x}\right|$.
The additional risk term is a directed drift term, where the uncertain
drift $\frac{\partial\Phi}{\partial x}(t,X_{t})$ scales with volatility
$\sigma$ and the coefficient $s_{\rho}$, which expresses risk aversion.
We want to emphasize that the nonlinear term $s_{\rho}\left|\sigma\frac{\partial\Phi}{\partial x}\right|$
is exactly the driver of a backwards stochastic differential equation
describing a coherent risk measure, also known as $g$-expectation.
Our approach is thus a constructive interpretation of the dynamic
risk measures discussed in \citet{Peng2004,Delong2013}. 

For absent risk, $s_{\rho}=0$, we obtain the classical \textendash{}
risk-neutral \textendash{} infinitesimal generator. Furthermore, if
$\sigma=0$, i.e., no randomness occurs in the model, the generator
reduces to a first order differential operator describing the dynamics
of a deterministic system, where risk does not apply.
\end{rem}

For random variables $Y$ of the form
\[
Y=\int_{t}^{T}c(s,X_{s})\,\mathrm{d}s+\Psi(X_{T}),
\]
where $X$ is an Itô diffusion process based on Brownian motion and
$c$, $\Psi$ are sufficiently smooth functions, the limit~\eqref{eq:CnRM}
exists as a consequence of Definition~\ref{def:divisibility} as
well as the arguments in the proof of Proposition~\ref{prop:RGen}
above.

\subsection{Dynamic programming}

This section introduces risk-averse dynamic equations using nested
risk measures. In what follows we consider the value function involving
nested risk measures defined by
\begin{equation}
V(t,x):=\rho^{t:T}\left(e^{-r(T-t)}\,\Psi(X_{T})\mid X_{t}=x\right).\label{eq:riskValue}
\end{equation}
Here, $r$ is a discount factor and $\Psi$ a terminal payoff function.
The structure of nested risk measures allows extending the dynamic
programming principle to the risk-averse setting.
\begin{lem}[Dynamic programming principle]
\label{lem:DPP}Let $(t,x)\in[0,T)\times\mathbb{R}$ and $\Delta t>0$,
then it holds that
\begin{equation}
V(t,x)=\rho^{t:t+\Delta t}\left(\left.e^{-r\Delta t}\,V(t+\Delta t,X_{t+\Delta t}\right|X_{t}=x\right).\label{eq:Principle}
\end{equation}
\end{lem}

\begin{proof}
By definition of the risk-averse value function~\eqref{eq:riskValue}
it holds that 
\[
V(t+\Delta t,X_{t+\Delta t})=\rho^{t+\Delta t:T}\left(e^{-r(T-t-\Delta t)}\Psi(X_{T})\mid X_{t+\Delta t}\right)
\]
and hence the construction of the nested risk measure gives
\begin{align*}
\rho^{t:t+\Delta t}\left(\left.e^{-r\Delta t}V(t+\Delta t,X_{t+\Delta t})\right|X_{t}=x\right) & =\rho^{t:T}\left(e^{-r(T-t)}\Psi(X_{T})\mid X_{t}=x\right),
\end{align*}
which shows the assertion.
\end{proof}
To derive the dynamic equations for $V$ we rearrange~\eqref{eq:Principle}
in the form
\begin{equation}
0=\frac{1}{\Delta t}\rho^{t:t+\Delta t}\left(\left.e^{-r\Delta t}V(t+\Delta t,X_{t+\Delta t})-V(t,x)\right|X_{t}=x\right)\label{eq:DPP2}
\end{equation}
and let $\Delta t\to0$. The following theorem employs the risk generator
to obtain dynamic equations for the risk-averse value function~\eqref{eq:riskValue}.
\begin{thm}
\label{thm:ValuePDE}The value function~\eqref{eq:riskValue} solves
the terminal value problem
\begin{align}
V_{t}(t,x)+b(t,x)V_{x}(t,x)+\frac{\sigma^{2}(t,x)}{2}V_{xx}(t,x)+s_{\rho}\left|\sigma(t,x)\cdot V_{x}(t,x)\right|-rV(t,x) & =0,\label{eq:5}\\
V(T,x) & =\Psi(x),\nonumber 
\end{align}
provided that $V\in C^{2}$ in a neighborhood of $(t,x)$ and $\sigma\cdot V_{x}$
satisfies the Hölder continuity assumption from Proposition~\ref{prop:RGen}.
\end{thm}

\begin{proof}
Let $(t,x)\in[0,T]\times\mathbb{R}$ be fixed. Similarly to the risk-neutral
case we define
\[
Y_{s}:=e^{-r(s-t)}V(s,X_{s}),\qquad s\geq t.
\]
By the Itô formula, the process $Y_{s}$ satisfies 
\begin{align*}
Y_{t+\Delta t}= & Y_{t}+\int_{t}^{t+\Delta t}e^{-r(s-t)}\left(V_{t}+b\cdot V_{x}+\frac{\sigma^{2}}{2}V_{xx}\right)(s,X_{s})-rV(s,X_{s})\,\mathrm{d}s\\
 & +\int_{t}^{t+\Delta t}e^{-r(s-t)}\sigma(s,X_{s})\cdot V_{x}(s,X_{s})\mathrm{d}W_{s}.
\end{align*}
As $\int_{t}^{t+\Delta t}\left(\sigma\cdot V_{x}\right)(t,x)\,\mathrm{d}W_{s}$
is normally distributed it follows from divisibility that 
\[
\lim_{\Delta t\downarrow0}\,\frac{1}{\Delta t}\rho^{t:t+\Delta t}\left(\left.\int_{t}^{t+\Delta t}\left(\sigma\cdot V_{x}\right)(t,x)\,\mathrm{d}W_{s}\right|X_{t}=x\right)=s_{\rho}\cdot\left|\sigma\cdot\partial_{x}V\right|(t,x)
\]
and thus following the lines of the proof of Proposition~\ref{prop:RGen}
shows
\begin{align*}
0 & =\lim_{\Delta t\downarrow0}\,\frac{1}{\Delta t}\rho^{t:t+\Delta t}\left(\left.Y_{t+\Delta t}-Y_{t}\right|X_{t}=x\right)\\
 & =\lim_{\Delta t\downarrow0}\,\frac{1}{\Delta t}\rho^{t:t+\Delta t}\left(\int_{t}^{t+\Delta t}e^{-r(s-t)}\left(V_{t}+b\cdot V_{x}+\frac{\sigma^{2}}{2}V_{xx}-rV\right)\mathrm{d}s+\int_{t}^{t+\Delta t}e^{-r(s-t)}\sigma\cdot V_{x}\mathrm{d}W_{s}\right)\\
 & =V_{t}(t,x)+b(t,x)\cdot V_{x}(t,x)+\frac{\sigma^{2}(t,x)}{2}V_{xx}(t,x)+s_{\rho}\left|\sigma(t,x)\cdot V_{x}(t,x)\right|-rV(t,x),
\end{align*}
demonstrating the assertion.
\end{proof}
\begin{rem}[Optimal controls]
\label{rem:HJB}The dynamic programming principle and Theorem~\ref{thm:ValuePDE}
are usually considered in an environment involving adapted controls~$u$.
This extends to the risk-averse setting as well. Here, we consider
the value function
\[
V(t,x):=\inf_{u}\,\rho^{t:T}\left(\int_{t}^{T}e^{-r(s-t)}c(s,X_{s}^{u},u_{s})\,\mathrm{d}s+e^{-r(T-t)}\Psi(X_{T}^{u})\right),
\]
where $X^{u}$ is a controlled diffusion process (see \citet{Fleming2006}).
Following the ideas in~\citet{Fleming2006} and using the structure
of nested risk measures as in the proof of Lemma~\ref{lem:DPP} we
may derive dynamic programming equations as
\[
V(t,x)=\inf_{u}\rho^{t:t+\Delta t}\left(\left.\int_{t}^{t+\Delta t}e^{-r(s-t)}c(s,X_{s}^{u},u_{s})\,\mathrm{d}s+e^{-r\Delta t}\,V(t+\Delta t,X_{t+\Delta t}\right|X_{t}=x\right).
\]
Moreover, following standard arguments, the Hamilton\textendash Jacobi\textendash Bellman
equation
\begin{align*}
\inf_{u}\left\{ V_{t}(\cdot)+b(\cdot,u)V_{x}(\cdot)+\frac{\sigma^{2}(\cdot,u)}{2}V_{xx}(\cdot)+s_{\rho}\left|\sigma(\cdot,u)\cdot V_{x}(\cdot)\right|-rV(\cdot)+c(\cdot,u)\right\}  & =0,\\
V(T,\cdot) & =\Psi(\cdot)
\end{align*}
characterizes the value function~$V$. We resume this discussion
in Section~\ref{sec:Merton} below.
\end{rem}

\section{\label{sec:EU}Pricing of options under risk}

The previous section discusses a discrete, risk-averse binomial option
pricing problem and studies the divisibility property of families
of risk measures. In this section we study the risk-averse value functions
of the limiting process of the binomial tree process, i.e., the geometric
Brownian motion. In the risk-averse setting we find again explicit
formulae. The resulting explicit pricing formulae lead us to interpret
risk aversion as dividend payments and to relate the risk level $s_{\rho}$
to the Sharpe ratio. Moreover, we establish the relationship between
divisibility and the convergence of binomial models under risk.

Consider a market with one riskless asset (a bond, e.g.) and a risky
asset, usually a stock. The return of the riskless asset is constant
and denoted by $r$. As usual in the classical Black\textendash Scholes
framework, the underlying stock $S$ is modeled by a geometric Brownian
motion following the stochastic differential equation
\begin{align}
\mathrm{d}S_{t} & =r\,S_{t}\,\mathrm{d}t+\sigma\,S_{t}\,\mathrm{d}W_{t}\label{eq:geom}
\end{align}
with initial value $S_{0}$.

\subsection{The risk-averse Black\textendash Scholes model for European options}

Similarly as above we distinguish the risk-averse value function 
\begin{equation}
V(t,x):=-\rho^{t:T}\left[-e^{-r(T-t)}\,\Psi(S_{T})\mid S_{t}=x\right]\label{eq:BSM-bid}
\end{equation}
for the \emph{bid price} and the corresponding value function for
the \emph{ask price} given by
\begin{equation}
\widetilde{V}(t,x):=\rho^{t:T}\left[e^{-r(T-t)}\,\Psi(S_{T})\mid S_{t}=x\right].\label{eq:BSM-ask}
\end{equation}
Notice that the discount rate $r$ is the same as in the dynamics~\eqref{eq:geom}
of the stock $S=(S_{t})_{t}$. In the risk-neutral setting the bid
and ask prices coincide.

Theorem~\ref{thm:ValuePDE} shows that the risk-averse value function~\eqref{eq:BSM-bid}
of the bid price satisfies the PDE 
\begin{align}
V_{t}(t,x)+r\,x\,V_{x}(t,x)+\frac{\sigma^{2}\,x^{2}}{2}V_{xx}(t,x)-s_{\rho}\cdot\left|\sigma\,x\cdot V_{x}(t,x)\right|-r\,V(t,x) & =0,\label{eq:PDE-bid}\\
V(T,x) & =\Psi(x),\nonumber 
\end{align}
the terminal value $\Psi(x)$ is the payoff function for either the
European put or call option. Similarly, the following PDE describes
the ask price $\widetilde{V}$, 
\begin{align}
\widetilde{V}_{t}(t,x)+r\,x\,\widetilde{V}_{x}(t,x)+\frac{\sigma^{2}\,x^{2}}{2}\widetilde{V}_{xx}(t,x)+s_{\rho}\cdot\left|\sigma\,x\cdot\widetilde{V}_{x}(t,x)\right|-r\,\widetilde{V}(t,x) & =0,\label{eq:PDE-ask}\\
\widetilde{V}(T,x) & =\Psi(x).\nonumber 
\end{align}
Notice that~\eqref{eq:PDE-bid} and~\eqref{eq:PDE-ask} differ only
in the sign of the nonlinear term, showing again that in the risk-neutral
setting (i.e., $s_{\rho}=0)$ the bid and ask prices coincide. We
have the following explicit solution of~\eqref{eq:PDE-bid} and~\eqref{eq:PDE-ask}
for the price of the call option.
\begin{prop}[Call option]
\label{prop:risky-BSM-2}Let $\Psi(x):=\max(x-K,0)$, define the
auxiliary functions (cf.\ \citet[Section 4.4]{Schachermayer+Delbaen-R})
\begin{equation}
d_{1}^{\pm}\coloneqq\frac{1}{\sigma\sqrt{T-t}}\cdot\left[\log\left(\frac{x}{K}\right)+\left(r\pm s_{\rho}\,\sigma+\frac{1}{2}\sigma^{2}\right)(T-t)\right],\qquad d_{2}^{\pm}\coloneqq d_{1}^{\pm}-\sigma\sqrt{T-t}\label{eq:7}
\end{equation}
and the value functions 
\begin{equation}
V^{\pm}(t,x):=xe^{\pm s_{\rho}\sigma(T-t)}\Phi(d_{1}^{\pm})-Ke^{-r(T-t)}\cdot\Phi(d_{2}^{\pm}),\label{eq:valueCall}
\end{equation}
where $\Phi$ denotes the cumulative distribution function of the
standard normal distribution. Then $V^{+}$ solves the risk-averse
Black\textendash Scholes PDE~\eqref{eq:PDE-ask} for the ask price,
while $V^{-}$ solves~\eqref{eq:PDE-bid}, the corresponding PDE
for the bid price; further, we have that $V^{-}\le V^{+}$. 
\end{prop}

We can solve the problem for the European put option similarly.
\begin{prop}[European Put option]
\label{prop:risky-BSM}Let $\Psi(x):=\max(K-x,0)$ and define the
value functions
\begin{equation}
V^{\mp}(t,x):=Ke^{-r(T-t)}\cdot\Phi(-d_{2}^{\mp})-xe^{\mp s_{\rho}\sigma(T-t)}\Phi(-d_{1}^{\mp}),\label{eq:valuePut}
\end{equation}
with $d_{1}^{\pm},d_{2}^{\pm}$ and $\Phi$ as in Proposition~\ref{prop:risky-BSM-2}.
Then $V^{-}$ solves the risk-averse Black\textendash Scholes PDE~\eqref{eq:PDE-ask}
and $V^{+}$ solves \eqref{eq:PDE-bid}, respectively. Note that $V^{+}\le V^{-}$.
\end{prop}

\begin{proof}
Plugging the value functions into the PDE~\eqref{eq:PDE-ask} and~\eqref{eq:PDE-bid}
shows the assertion.
\end{proof}

\subsection{Rationale of risk aversion in the new formulae}

\subsubsection{On the nature of the risk level $s_{\rho}$}

The Propositions~\ref{prop:risky-BSM-2} and~\ref{prop:risky-BSM}
show that the value function for the risk-averse European option pricing
problem can be identified with the risk-neutral problem, where the
stock pays dividends. In case of the bid price of a European call
option the risk dividend is $s_{\rho}\,\sigma$. Similarly, the dividend
for the bid price for a European put option is $-s_{\rho}\,\sigma$,
thus negative. For an increasing risk aversion coefficients $s_{\rho}$,
the bid price for the put and the call price decrease. This monotonicity
reverses for the ask price. It is important to note that stocks do
not pay negative dividends and thus negative risk dividends may be
interpreted as a premium for holding the option rather than a dividend
payment from the underlying stock. 

The value functions~\eqref{eq:valueCall} and~\eqref{eq:valuePut}
can also be interpreted within the framework of the Garman\textendash Kohlhagen
model on foreign exchange options. In this sense $s_{\rho}\,\sigma$
corresponds to the interest rate in the foreign currency. We illustrate
this for the bid price of a European call option. Recall that the
value of a call option into a foreign currency with interest rate
$r_{f}$ satisfies

\begin{align*}
V^{GK}(t,x) & :=x\,e^{-r_{f}(T-t)}\cdot\Phi(d_{1}^{-})-Ke^{-r_{d}(T-t)}\cdot\Phi(d_{2}^{-}),
\end{align*}
where $r_{d}$ is the interest in the domestic currency and
\[
d_{1}^{-}\coloneqq\frac{1}{\sigma\sqrt{T-t}}\cdot\left[\log\left(\frac{x}{K}\right)+\left(r_{d}-r_{f}+\frac{1}{2}\sigma^{2}\right)(T-t)\right],\qquad d_{2}^{-}\coloneqq d_{1}^{-}-\sigma\sqrt{T-t}.
\]
Comparing with Equation~\eqref{eq:valueCall} we notice that $r$
can be identified with the domestic interest rate $r_{d}$ ($r_{d}=r$)
and $s_{\rho}\,\sigma$ with the foreign interest rate $r_{f}$ ($r_{f}=s_{\rho}\,\sigma$),
which bears the risk. The option price $V^{GK}$ represents the value
in domestic currency of a call option. Risk aversion is encoded in
the underlying, which is the foreign currency.

A risk-averse investor assumes a return $\mu_{\mathit{averse}}$ for
the underlying asset. Subsection~\ref{subsec:Illustration} below
then identifies $s_{\rho}\,\sigma$ with $r_{d}-\mu_{\mathit{averse}}$.
Comparing with the Garman\textendash Kohlhagen model we observe that
the foreign currency $r_{f}$ encodes the spread between the risk-neutral
and the risk-averse setting.

\subsubsection{Illustration of the risk level $s_{\rho}$ \label{subsec:Illustration}}

Figure~\ref{fig:prices} displays risk-averse prices for put and
call options  from buyer's and seller's perspectives. As a reference
we include the risk-neutral Black\textendash Scholes price as well.
For this illustration we choose $T=1$ with strike $K=1.2$, the interest
rate is $r=3\,\%$ and the volatility is $\sigma=15\,\%$. Figure~\ref{fig:spread}
exhibits the bid-ask spread, which is present in the risk-averse situation.
\begin{figure}[H]
\centering{}\subfloat[Call prices]{\begin{centering}
\includegraphics[viewport=10bp 5bp 420bp 305bp,clip,width=0.49\textwidth]{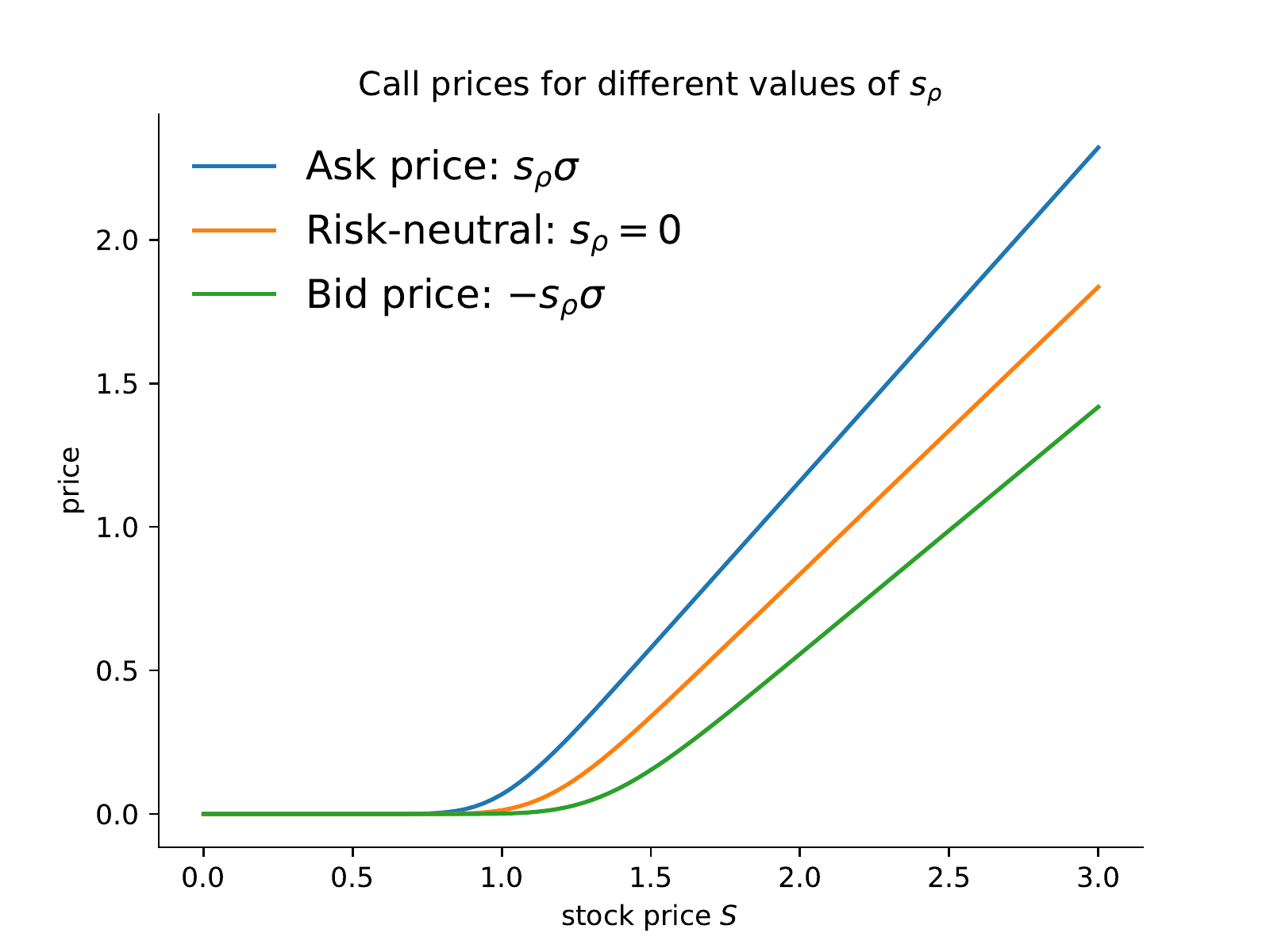}
\par\end{centering}
}\subfloat[Put prices]{\begin{centering}
\includegraphics[viewport=10bp 5bp 420bp 307bp,clip,width=0.49\textwidth]{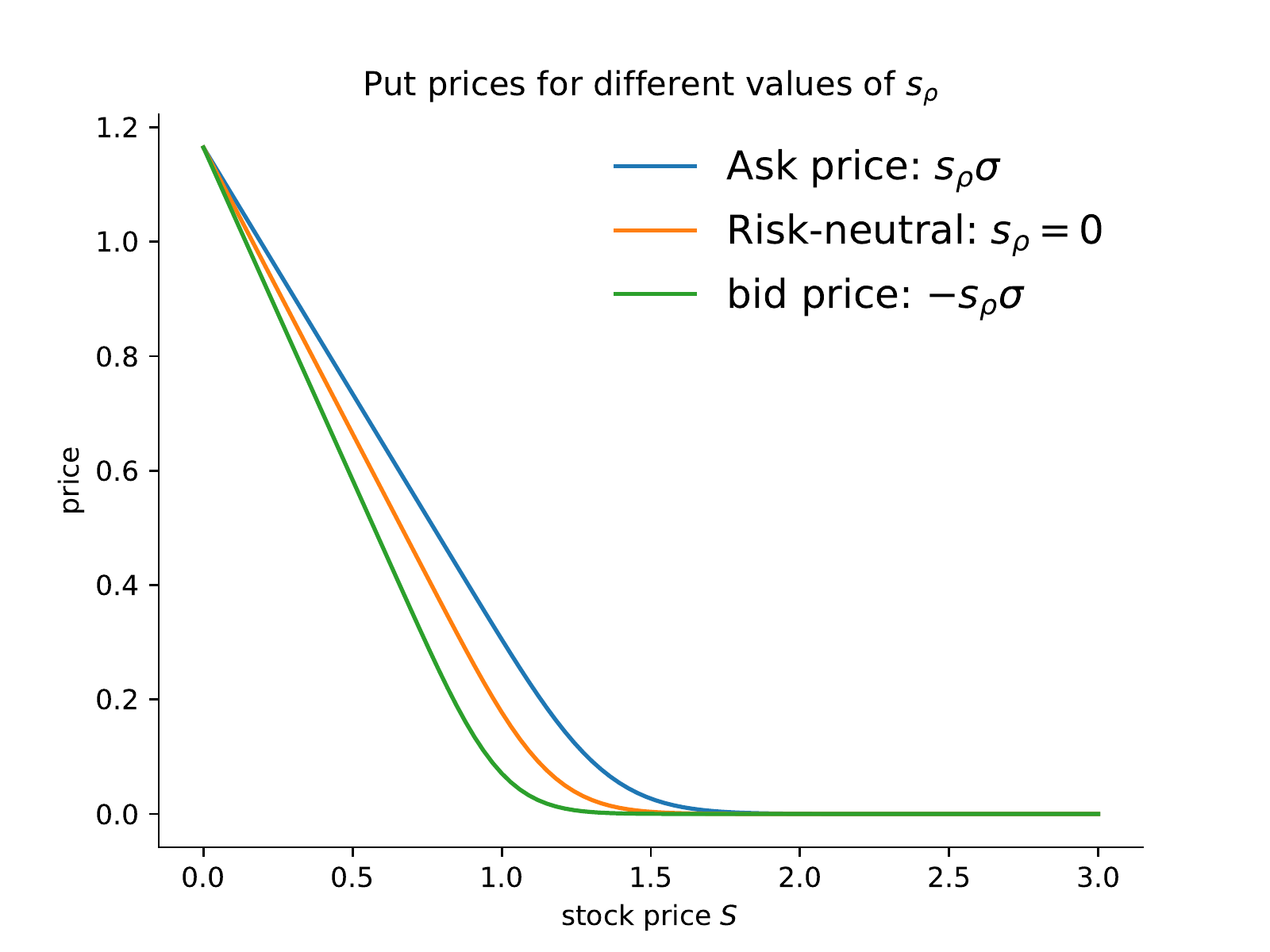}
\par\end{centering}
}\caption{\label{fig:prices}European option prices for different risk levels}
\end{figure}

\subsubsection{Discussion of the risk level $s_{\rho}$}

The Sharpe ratio is 
\[
\frac{\mu-r}{\sigma},
\]
where $\mu$ is the mean return of an asset with volatility $\sigma$
and $r$ is the risk free interest rate. Comparing units in~\eqref{eq:7}
we see that $s_{\rho}\,\sigma$ is an interest rate and hence $s_{\rho}$
has unit
\[
\frac{\mathit{interest}}{\mathit{volatility}},
\]
the same unit as the Sharpe ratio. 

To explore that the risk-aversion coefficient $s_{\rho}$ has the
structure of a Sharpe ratio denote by $\mu_{\mathit{averse}}$ the
mean return a risk-averse investor expects. Depending on the sign
we may equate 
\begin{equation}
\frac{\mu_{\mathit{averse}}-r}{\sigma}=\pm s_{\rho}\label{eq:26}
\end{equation}
with $s_{\rho}$ as in~\eqref{eq:17-1} above. The parallel shift
\[
r-\mu_{\mathit{averse}}=\pm s_{\rho}\cdot\sigma
\]
over the risk free interest derived from~\eqref{eq:26} is known
as \emph{Z-spread} in economics.
\noindent \begin{center}
\begin{figure}[H]
\centering{}\subfloat[European call option]{\begin{centering}
\includegraphics[viewport=0bp 0bp 461bp 307bp,clip,width=0.49\textwidth]{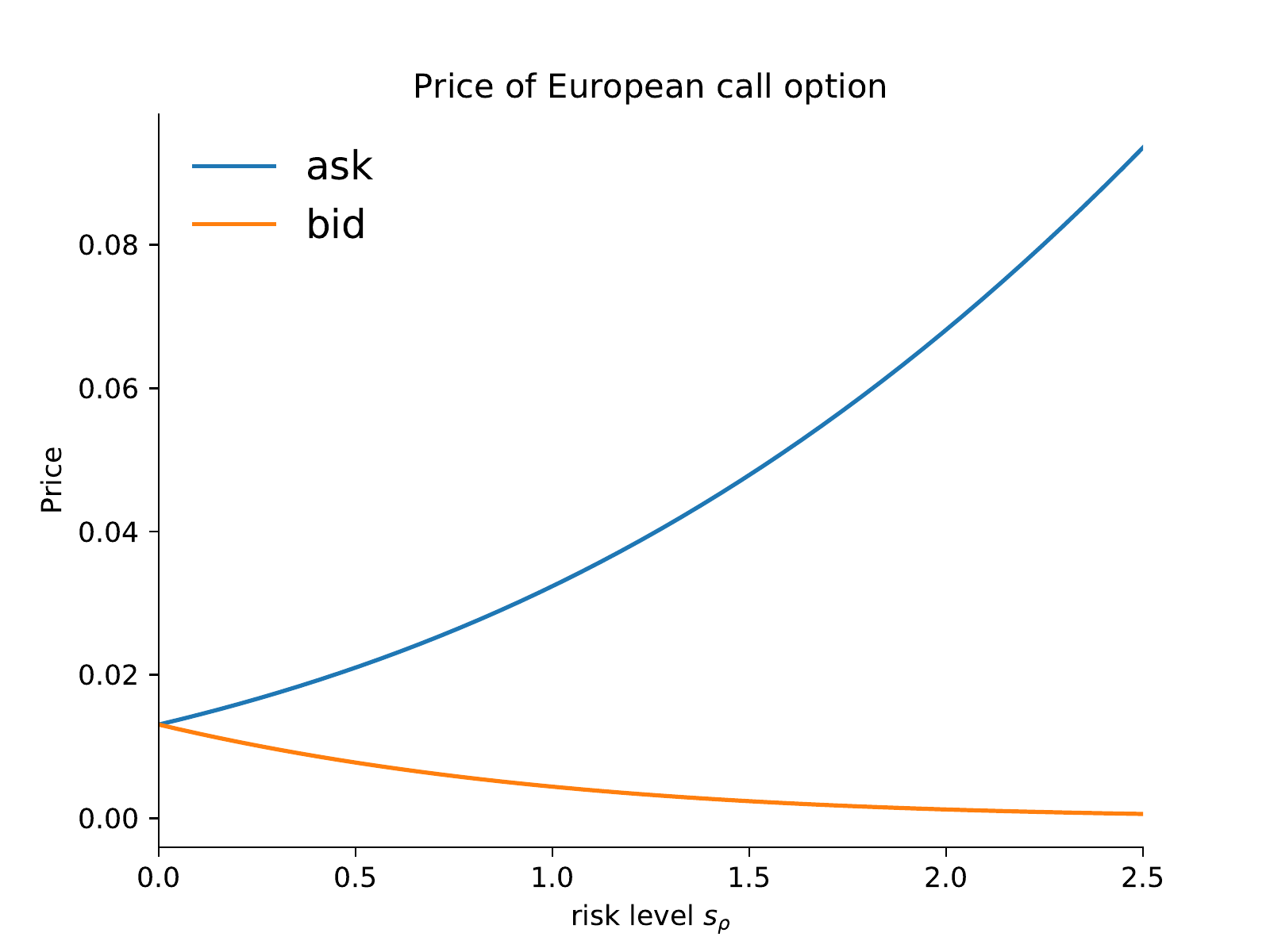}
\par\end{centering}
}\subfloat[European put option]{\begin{centering}
\includegraphics[viewport=0bp 0bp 461bp 307bp,clip,width=0.49\textwidth]{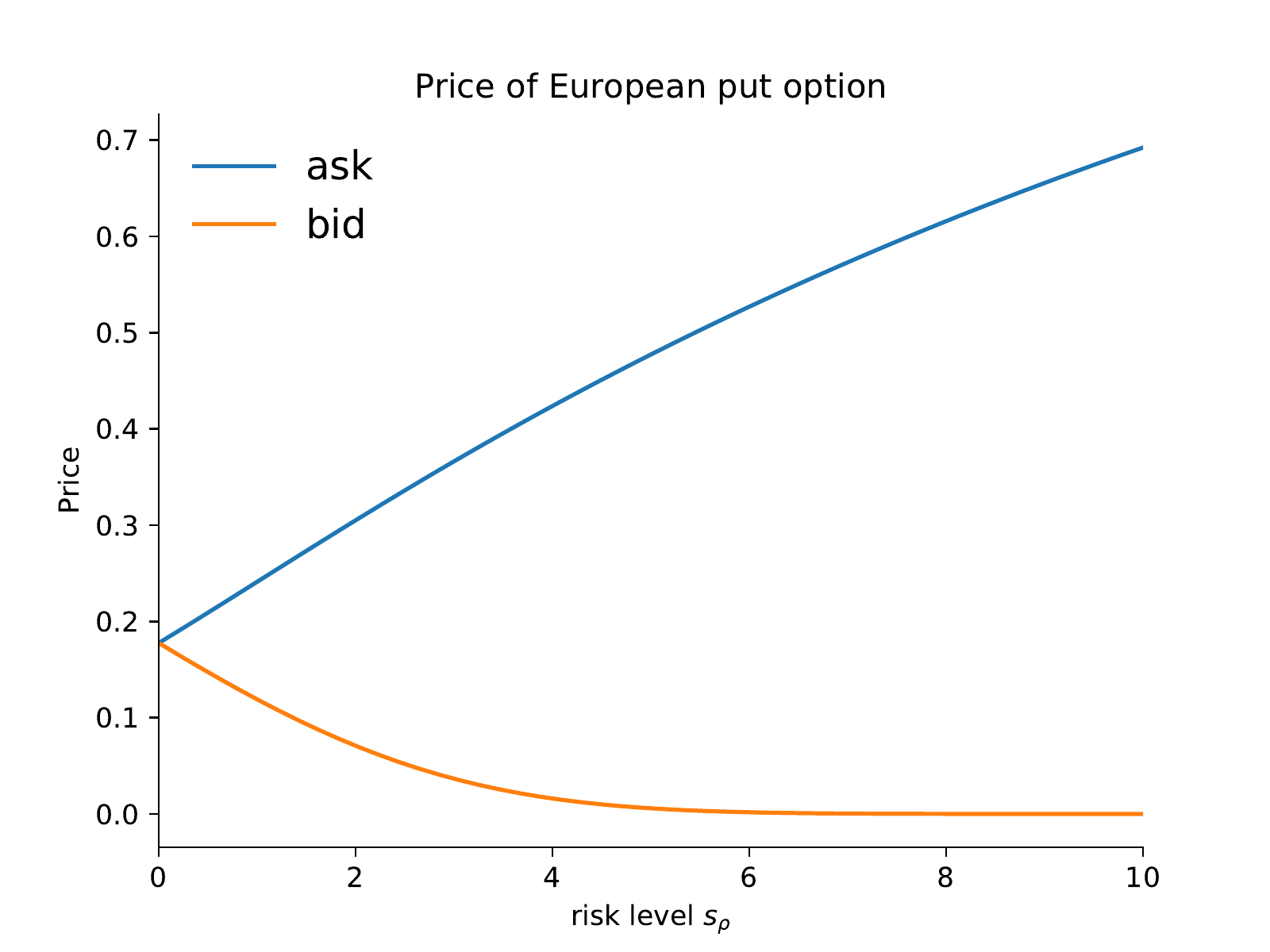}
\par\end{centering}
}\caption{\label{fig:spread}The bid-ask spread for varying risk level $s_{\rho}$}
\end{figure}
\par\end{center}
\begin{rem}
\noindent Figure~\ref{fig:spread} (as well as Figure~\ref{fig:3-1}
below) reveals opposite slopes of the bid and ask price at $s_{\rho}=0$,
the Black\textendash Scholes price. This reflects the opposing risk
assessment of the buying and selling investor at comparable risk aversion
coefficients. The value function~\eqref{eq:valueCall} is indeed
differentiable at $s_{\rho}=0$ and the sensitivity with respect to
the risk dividend $s_{\rho}\,\sigma$ relates to the classical Greek~$\varepsilon$
(or~$\psi$) for dividend paying models. 
\end{rem}

\subsection{\label{subsec:consistency}Consistency with discrete models}

We return to the binomial model with risk-averse probabilities from
Remark~\ref{rem:risk-binom}. The preceding discussions on divisibility
and the risk generator show that the risk level $\beta$ for the mean
semi-deviation risk measure needs to be proportional to 
\[
\sqrt{\Delta t}.
\]
Further recall the risk-neutral probabilities 
\[
p=\frac{e^{r\Delta t}-e^{-\sigma\sqrt{\Delta t}}}{e^{\sigma\sqrt{\Delta t}}-e^{-\sigma\sqrt{\Delta t}}}=\frac{1}{2}+\left(\frac{r}{2\sigma}-\frac{\sigma}{4}\right)\sqrt{\Delta t}+o(\Delta t)
\]
and hence the risk-averse probabilities in~\eqref{eq:15} satisfy
\[
\widetilde{p}=p(1-\beta\sqrt{\Delta t}(1-p))=\frac{1}{2}+\left(\frac{r-\frac{\beta\,\sigma}{2}}{2\sigma}-\frac{\sigma}{4}\right)\sqrt{\Delta t}+o(\Delta t).
\]
Thus replacing the interest rate $r$ by $r-\frac{\beta\,\sigma}{2}$
shows that under the nested mean semi-deviation the distribution for
the stock $S_{t}$ is
\[
S_{t}=S_{0}\exp\left\{ t\left(r-\frac{\beta\,\sigma}{2}-\frac{\sigma^{2}}{2}\right)+\sigma\,W_{t}\right\} .
\]

Recall from Lemma~\ref{lem:gauss} that $s_{\rho}=\frac{\beta}{\sqrt{2\pi}}$
for the mean semi-deviation of order $p=1$. However, the binomial
model converges to a process with dividends $\frac{\beta}{2}\sigma>s_{\rho}\,\sigma$.
The deviating scaling factors are in line with the discontinuity of
coherent risk measures with respect to convergence in distribution,
described in \citet[Theorem~4.1]{BaeuerleMueller2006-R}. The discussion
shows that adapting the risk level $\beta$ of the nested mean semi-deviation
leads to a well-defined limit in continuous time. 

In general, one may not expect that nesting conditional risk measures
leads to a well-defined risk measure in continuous time. \citet{Xin2011}
first observed that naively nesting the conditional Average Value-at-Risk
leads to an exponentially increasing upper bound and \citet{PichlerSchlotter2019}
extend this result to more general risk measures (see also \citet{Pichler2017-R}
for a collection of related inequalities).

The following proposition extends the discussion of the nested mean
semi-deviation to more general risk measures and provides the theoretical
connection between divisibility and convergence of risk-averse option
pricing models. 
\begin{prop}
Denote by $S^{n}$ the $n$-period binomial tree model~\eqref{eq:Binomial}
converging to a geometric Brownian motion for $n\to\infty$. Then
the risk-averse binomial model in Remark~\ref{rem:risk-binom} converges
if the family of nested risk measures is divisible.
\end{prop}

\begin{proof}
Let $(\rho_{\Delta t})_{\Delta t}$ be a divisible family of risk
measures and denote by $X=(X_{t})_{t}$ the geometric Brownian motion.
As $X_{0}=S_{0}^{n}$ for all $n$ we have the following inequality,
\[
\lim_{n\to\infty}\rho_{\Delta t}\left(S_{\Delta t}^{n}-S_{0}^{n}\right)\leq\lim_{n\to\infty}\rho_{\Delta t}\left(S_{\Delta t}^{n}-X_{\Delta t}\right)+\rho_{\Delta t}\left(X_{\Delta t}-X_{0}\right).
\]
Because $(\rho_{\Delta t})_{\Delta t}$ is a divisible family of risk
measures Proposition~\ref{prop:RGen} shows that
\[
\rho_{\Delta t}\left(X_{\Delta t}-X_{0}\right)=c_{\rho}\cdot\Delta t+o(\Delta t).
\]
For the first term notice that $(S_{\Delta t}^{n}-X_{\Delta t})_{n}$
tends to zero in distribution and hence also converges in probability.
Moreover, $\left(S_{\Delta t}^{n}-X_{\Delta t}\right)_{n}$ is uniformly
bounded in $L^{p}$ and hence with divisibility and dominated convergence
\[
\lim_{n\to\infty}\rho_{\Delta t}\left(S_{\Delta t}^{n}-X_{\Delta t}\right)=0.
\]
It follows that
\[
\lim_{n\to\infty}\rho_{\Delta t}\left(S_{\Delta t}^{n}-S_{0}^{n}\right)=c_{\rho}\cdot\Delta t+o(\Delta t),
\]
which implies the existence of the limit of risk-averse binomial models
as in Remark~\ref{rem:risk-binom}.
\end{proof}

\subsection{\label{sec:US}Pricing of American options under risk}

The Black\textendash Scholes model allows explicit formulae for European
option prices in in the risk-averse setting. This is surprising given
the initial nonlinear PDE formulation in~\eqref{eq:PDE-bid} and~\eqref{eq:PDE-ask}.
Similarly we may reformulate the risk-averse American option pricing
problem and in what follows we introduce the risk-averse optimal stopping
problem for American put options and introduce the value functions.

Again we assume that the stock $S$ follows the geometric Brownian
motion~\eqref{eq:geom}. Here, the risk-averse bid price of an American
option is given by $\sup_{\tau\in[0,T]}\,-\rho^{0:\tau}\left[-e^{-r\tau}\,\Psi(S_{\tau})\right]$,
where $\Psi(\cdot)$ is the payoff function and the supremum is among
all stopping times with $\tau\in[0,T]$. The ask price is given by
$\sup_{\tau\in[0,T]}\,\rho^{0:\tau}\left[e^{-r\tau}\,\Psi(S_{\tau})\right]$.
We can further define the value functions
\[
V(t,x):=\sup_{\tau\in[t,T]}\,-\rho^{t:\tau}\left[-e^{-r(\tau-t)}\,\Psi(S_{\tau})\mid S_{t}=x\right]
\]
for the bid price and
\[
\widetilde{V}(t,x):=\sup_{\tau\in[t,T]}\,\rho^{t:\tau}\left[e^{-r(\tau-t)}\,\Psi(S_{\tau})\mid S_{t}=x\right]
\]
for the ask price. For brevity we only discuss the bid price for American
put options, the arguments for the ask price are analogous. By informally
extending the arguments from the risk-neutral setting to the risk-averse
setting we obtain the free boundary problem 
\begin{eqnarray}
V_{t}(t,x)+rxV_{x}(t,x)+\frac{\sigma^{2}x^{2}}{2}V_{xx}(t,x)-s_{\rho}\sigma x\left|V_{x}\right|= & rV(t,x) & \text{for }x\geq L(t),\label{eq: freeDiffu}\\
V(t,x)= & (K-x)_{+} & \text{for }0\leq x<L(t),\\
V_{x}(t,x)= & -1 & \text{for }x=L(t),\label{eq:freeBee}\\
V(T,x)= & (K-x)_{+}\nonumber \\
L(T)= & K\nonumber \\
\lim_{x\to\infty}V(t,x)= & 0 & \text{for }0\leq t\leq T\label{eq:8}
\end{eqnarray}
for the optimal exercise boundary $t\mapsto L(t)$. For an overview
on American options and free boundary problems in general we refer
to \citet{Peskier2006}. The following result follows with standard
arguments for American options. 
\begin{thm}
The value function 
\begin{equation}
V(t,x)=\sup_{\tau\in[t,T]}\,-\rho^{t:\tau}\left[-e^{-r(\tau-t)}(K-S_{\tau})_{+}\mid S_{t}=x\right]\label{eq:US-value}
\end{equation}
solves the free boundary problem~\eqref{eq: freeDiffu}\textendash \eqref{eq:8}. 
\end{thm}

Similarly to European options, risk-aversion reduces to a modification
of the drift term and the standard American put option model applies
for an underlying stock with risk dividends. To this end notice that
\begin{align*}
V_{t}(t,x)+rxV_{x}(t,x) & +\frac{\sigma^{2}x^{2}}{2}V_{xx}(t,x)-s_{\rho}\sigma x\left|V_{x}\right|\\
=\inf_{y\in[-1,1]} & \left\{ V_{t}(t,x)+\left(r-s_{\rho}\sigma y\right)xV_{x}(t,x)+\frac{\sigma^{2}x^{2}}{2}V_{xx}(t,x)\right\} 
\end{align*}
provided that $x\geq L(t)$. The American option is not exercised
and the same arguments as for the European options show that the infimum
over all constraints is attained at $y=-1$. The equation~\eqref{eq: freeDiffu}
is thus equal to
\[
V_{t}(t,x)+\left(r+s_{\rho}\sigma\right)xV_{x}(t,x)+\frac{\sigma^{2}x^{2}}{2}V_{xx}(t,x)=r\,V(t,x)\qquad\text{for }x\geq L(t).
\]
Consequently we deduce that the value function
\[
V(t,x):=\sup_{\tau\in[t,T]}\,\E\left[e^{-r(\tau-t)}\Psi\left(S_{\tau}\right)\mid S_{t}=x\right]
\]
solves the free boundary problem~\eqref{eq: freeDiffu}\textendash \eqref{eq:8},
where the state process is given by
\begin{align*}
\mathrm{d}S_{s} & =\left(r+s_{\rho}\,\sigma\right)S_{s}\,\mathrm{d}s+\sigma\,S_{s}\,\mathrm{d}W_{s}
\end{align*}
for a risk-loaded interest interest rate.

\subsection*{Numerical illustration}

Consider the geometric Brownian motion
\begin{align*}
\mathrm{d}S_{t} & =0.03S_{t}\,\mathrm{d}t+0.15S_{t}\,\mathrm{d}W_{t},\qquad0<t\leq1,\\
S_{0} & =1.
\end{align*}
The strike price in the next Figure~\ref{fig:3} is $K=1$. We consider
the optimal stopping region for different risk levels $s_{\rho}$.
A risk-averse option buyer (bid price) would generally exercise earlier,
he accepts less profits due to his risk aversion. Compared with the
risk-neutral investor, the risk aware option buyer prefers exercising
prematurely rather than delayed exercise.

The reverse is true for the option holder (ask price), where the investor
waits longer.

\begin{figure}[H]
\centering{}\subfloat[bid price]{\begin{centering}
\includegraphics[viewport=10bp 0bp 420bp 307bp,clip,width=0.49\textwidth]{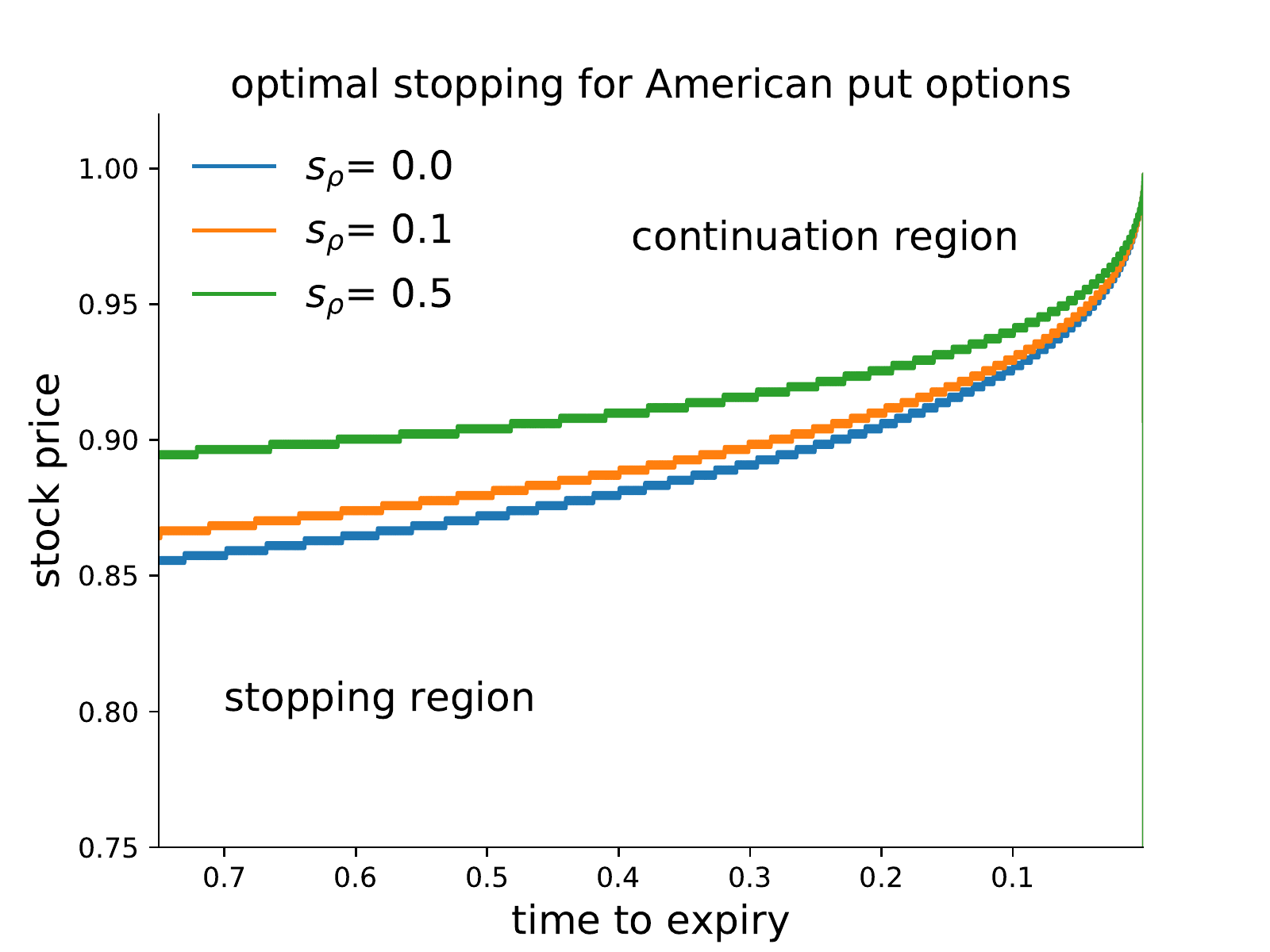}
\par\end{centering}
}\subfloat[ask price]{\begin{centering}
\includegraphics[viewport=10bp 0bp 420bp 307bp,clip,width=0.49\textwidth]{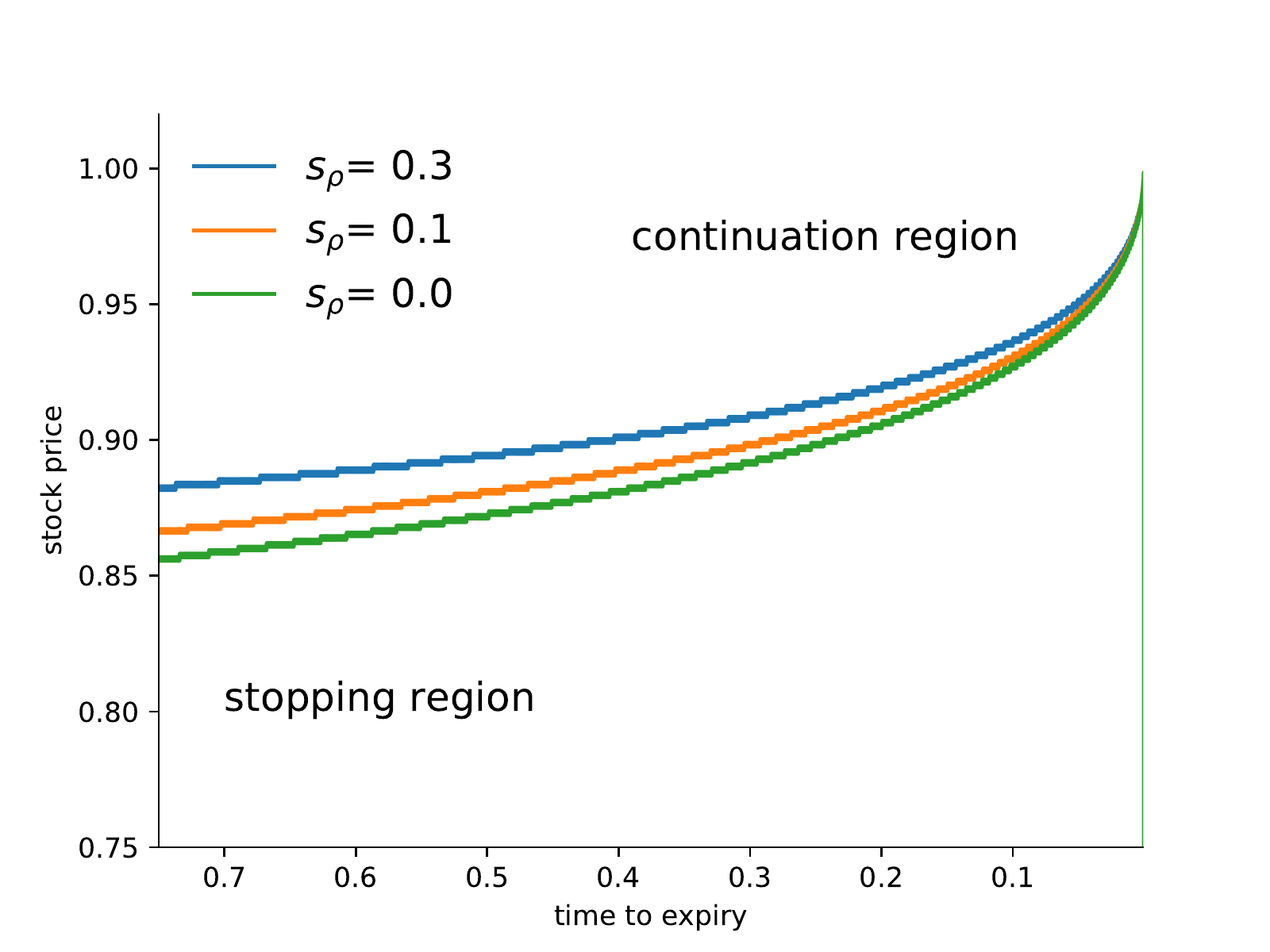}
\par\end{centering}
}\caption{\label{fig:3}optimal stopping regions for put options}
\end{figure}

In the risk-neutral case it is never optimal to exercise an American
call option before expiry. However, this is only the case if the interest
rate exceeds the dividends of the underlying asset (see, for instance,
\citet[Chapter 8.5]{Shreve2010} for details). As nested risk measures
modify the interest rate it may be optimal to exercise the call option
early. Figure~\ref{fig:stopCall} shows the optimal exercise boundary
for the risk-averse call option with strike $K=1$ and initial value
$S_{0}=1$.

\begin{figure}[H]
\centering{}\includegraphics[viewport=0bp 0bp 413bp 307bp,clip,scale=0.5]{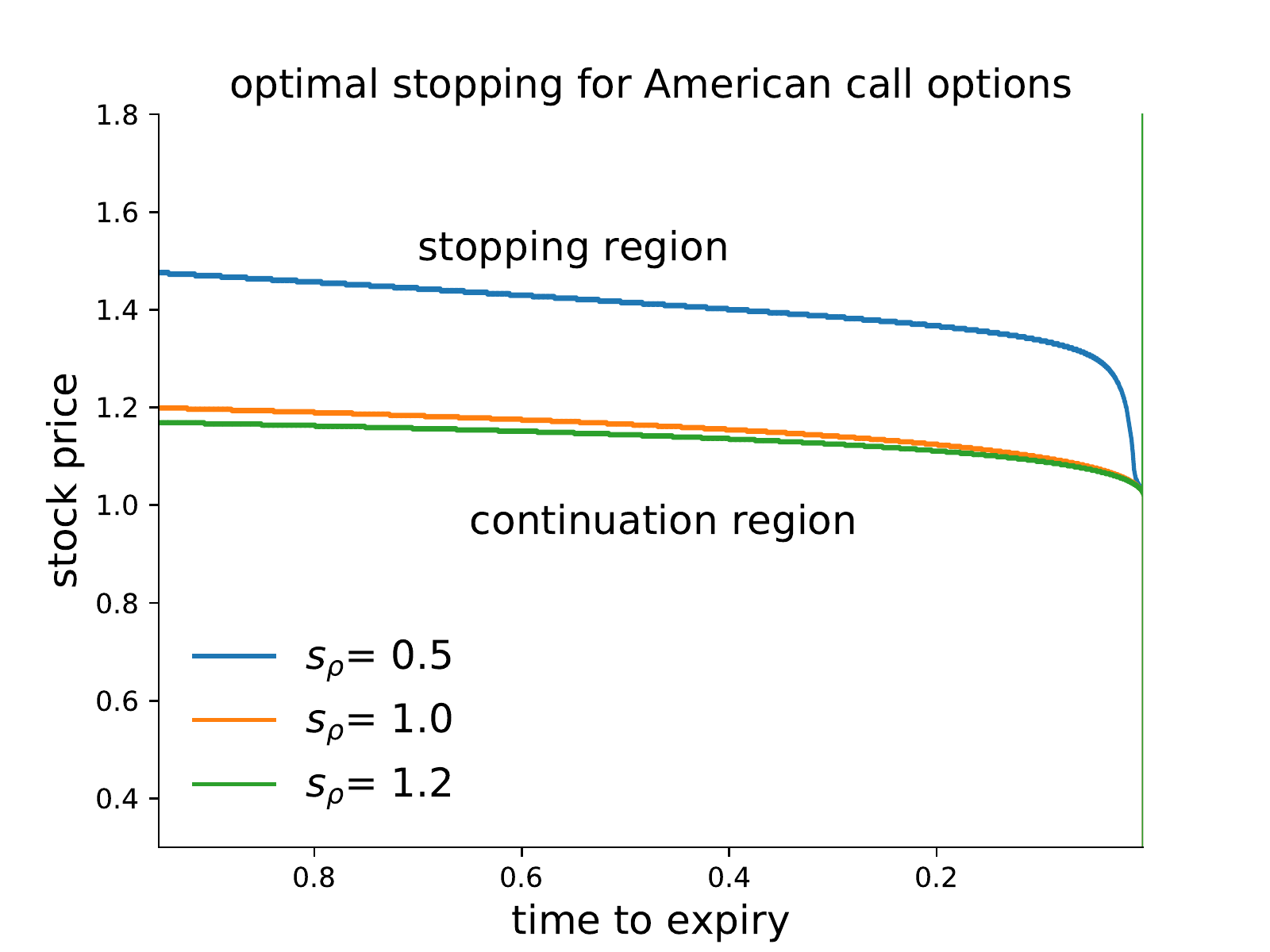}\caption{\label{fig:stopCall}optimal stopping regions for different risk-levels
(call option)}
\end{figure}

Below we show the bid-ask spread for American options.

\begin{figure}[H]
\centering{}\subfloat[American call option]{\begin{centering}
\includegraphics[viewport=10bp 0bp 420bp 307bp,clip,width=0.49\textwidth]{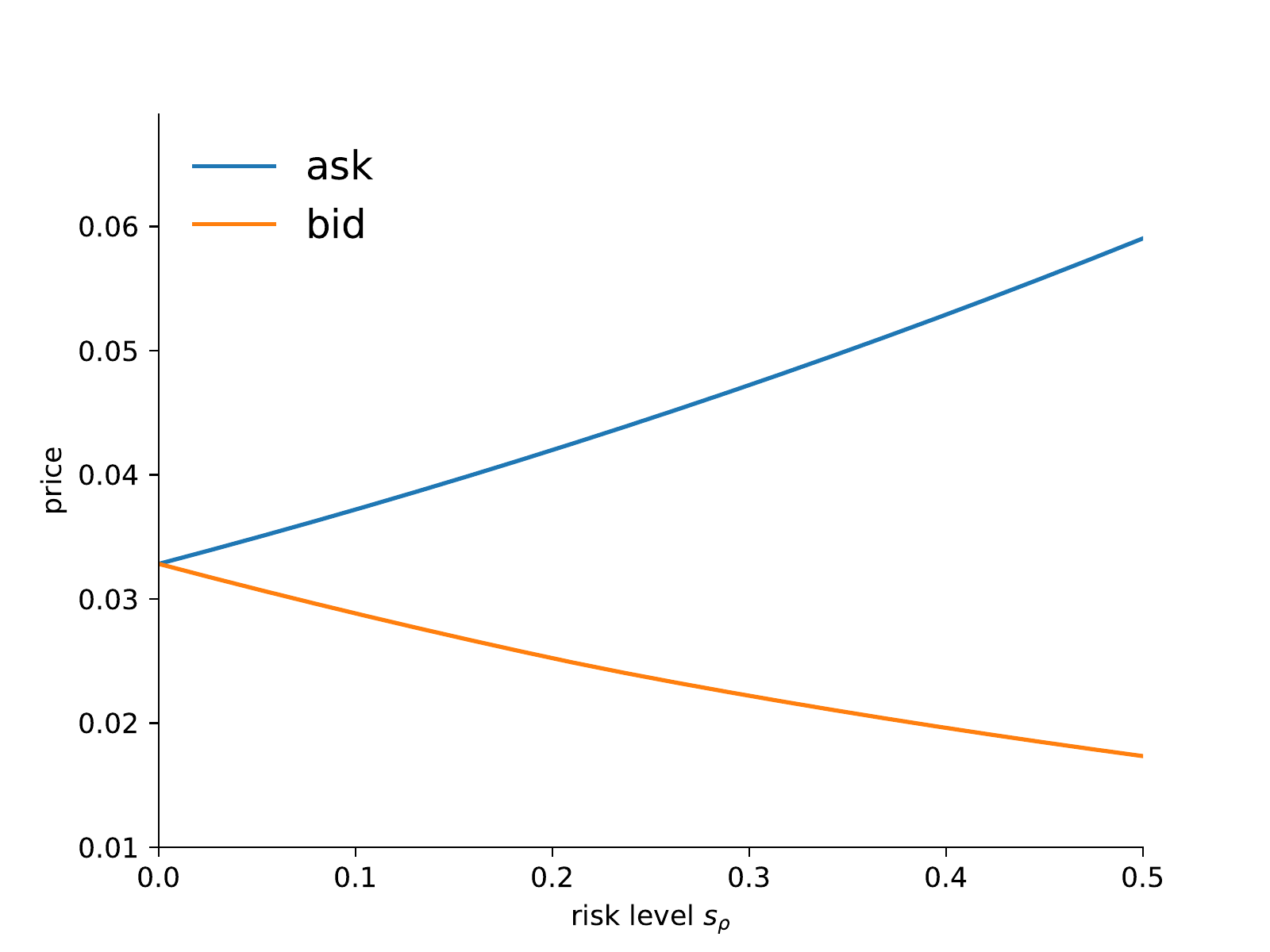}
\par\end{centering}
}\subfloat[American put option]{\begin{centering}
\includegraphics[viewport=10bp 0bp 420bp 307bp,clip,width=0.49\textwidth]{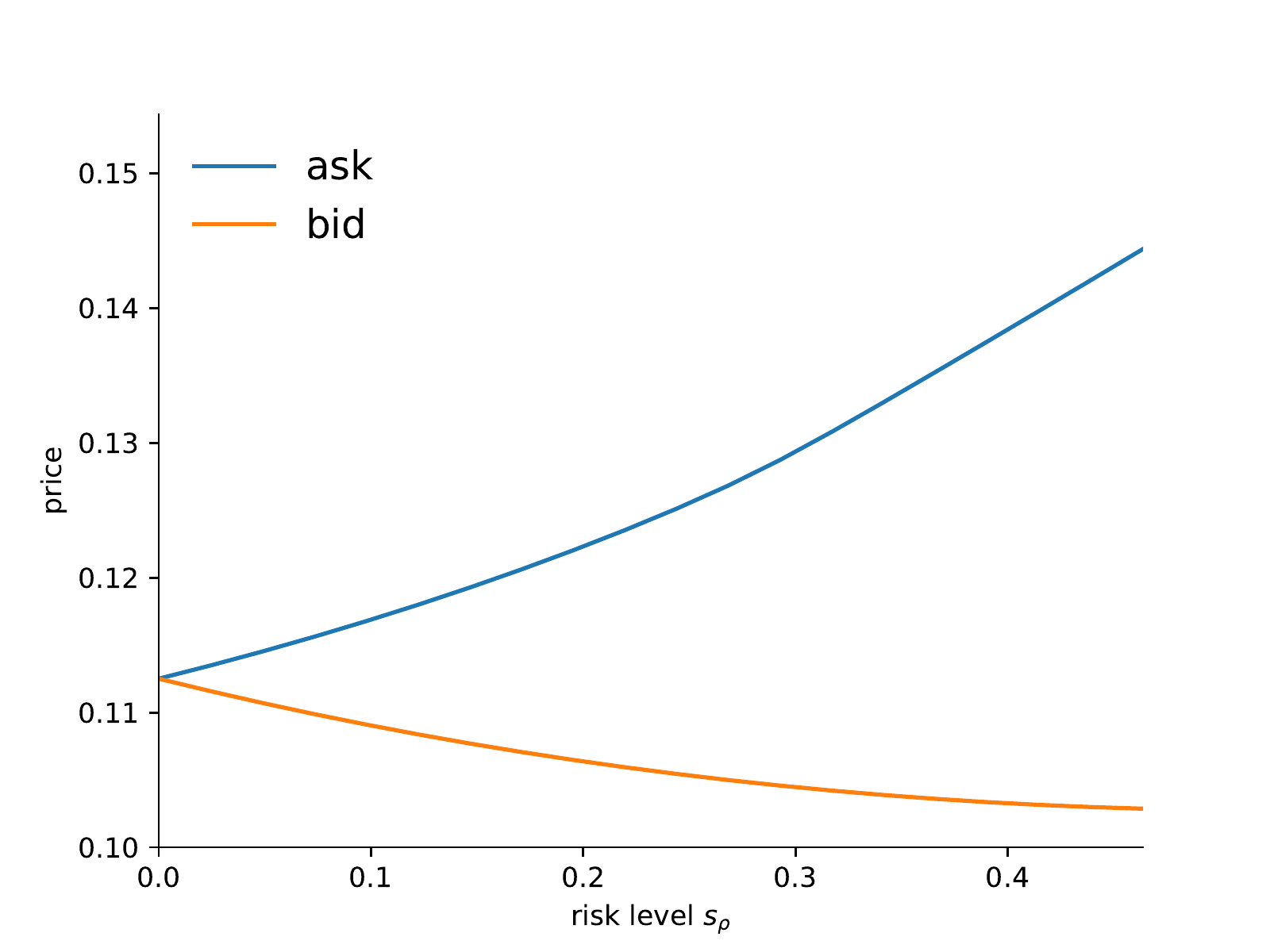}
\par\end{centering}
}\caption{\label{fig:3-1}risk-averse American option values}
\end{figure}

\section{The Merton problem\label{sec:Merton}}

The preceding sections demonstrate that classical option pricing models
generalize naturally to a risk-averse setting by employing nested
risk measures. In what follows we demonstrate that the classical Merton
problem, which allows an explicit solution in specific situations,
as well allows extending to the risk-averse situation. 

Consider a risk-less bond $B$ satisfying the ordinary differential
equation $\mathrm{d}B_{t}=r\,B_{t}\,\mathrm{d}t$ and a risky asset~$S$
driven by the stochastic differential equation 
\[
\mathrm{d}S_{t}=\mu S_{t}\,\mathrm{d}t+\sigma S_{t}\,\mathrm{d}W_{t}.
\]
We are interested in the optimal fraction $\pi_{t}$ of the total
wealth $w_{t}$ one should invest in the risky asset. The wealth process
is 
\[
\mathrm{d}w_{t}=\left[\left(\pi_{t}\mu+(1-\pi_{t})r\right)w_{t}-c_{t}\right]\mathrm{d}t+\pi_{t}\,\sigma\,w_{t}\,\mathrm{d}W_{t},
\]
where $c_{t}$ is the rate of consumption. Following Merton we employ
the power utility function $u(x)=\frac{x^{1-\gamma}}{1-\gamma}$ with
parameter $\gamma\ge0$ and $\gamma\neq1$ and consider the risk-averse
objective function
\begin{equation}
R(t,x):=\sup_{\pi,c}\,-\rho^{t:T}\left(-\int_{t}^{T}u(c_{s})\mathrm{d}s-\epsilon^{\gamma}\,u(w_{T})\mid w_{t}=x\right),\label{eq:optimalControl}
\end{equation}
where $\epsilon$ parameterizes the desired payout at terminal time.
Surprisingly, $R$ has a closed form solution and the optimal portfolio
allocation of the risk averse investor is 
\[
\pi^{*}=\max\left(\frac{\mu-r-s_{\rho}\,\sigma}{\sigma^{2}\,\gamma},0\right).
\]
We observe again that risk aversion leads to a modified drift term
$r+s_{\rho}\sigma$ in place of $r$. The optimal portfolio allocation
$\pi^{*}$ is a decreasing function of $s_{\rho}$. This is in line
with the usual economic perception, as increasing risk-aversion corresponds
to less investments into the risky asset. The optimal consumption
is given by
\[
c_{t}^{*}(x)=\frac{x\,\nu}{1+(\nu\,\epsilon-1)e^{-\nu(T-t)}},
\]
where $\nu$ is a constant depending on the model parameters. Consumption
generally increases with risk aversion as the value of immediate consumption
offsets the present value of uncertain wealth in the future.

In Remark~\ref{rem:HJB} we formally extended the results of Proposition~\ref{prop:RGen}
to objective functions of the form~\eqref{eq:optimalControl}. Based
on this we now consider the Hamilton\textendash Jacobi\textendash Bellman
equation 
\begin{align}
0 & =\max_{\pi,c}\,\left[R_{t}+\left[\left(\pi_{t}\,\mu+(1-\pi_{t})r\right)x-c_{t}\right]R_{x}+\frac{\sigma^{2}\pi^{2}x^{2}}{2}R_{xx}+u(c_{t})-s_{\rho}\left|\sigma\,\pi_{t}\,x\,R_{x}\right|\right]\label{eq:HJB}
\end{align}
with terminal condition $R(T,x)=\frac{\epsilon^{\gamma}}{1-\gamma}x^{1-\gamma}$.
In what follows we derive the optimal value function $R$ and verify
the optimal portfolio allocation $\pi^{*}$ and optimal consumption
$c^{*}$ given above.

The Hamilton\textendash Jacobi\textendash Bellman equation~\eqref{eq:HJB}
allows for explicit optimal controls outlined in the following proposition.
\begin{prop}
In the risk-averse setting, the optimal controls are given by
\begin{align*}
\pi_{t}^{*} & (x)=-\frac{(\mu-r)R_{x}}{\sigma^{2}xR_{xx}}+\frac{s_{\rho}\sigma\left|R_{x}\right|}{\sigma^{2}xR_{xx}},\qquad c_{t}^{*}(x)=R_{x}^{-\frac{1}{\gamma}}.
\end{align*}
The Hamilton-Jacobi-Bellman equation~\eqref{eq:HJB} rewrites as
\begin{align}
0 & =R_{t}-\frac{\left((\mu-r)^{2}+s_{\rho}^{2}\sigma^{2}\right)R_{x}^{2}}{2\sigma^{2}R_{xx}}+\frac{s_{\rho}R_{x}\left|R_{x}\right|}{\sigma R_{xx}}+rxR_{x}+\frac{\gamma}{1-\gamma}R_{x}^{\frac{\gamma-1}{\gamma}},\label{eq:HJB-2}\\
R(T,x) & =\frac{\epsilon^{\gamma}}{1-\gamma}x^{1-\gamma}.\nonumber 
\end{align}
\end{prop}

The preceding proposition derives first order conditions for the fraction
$\pi_{t}^{*}$ and consumption rate $c_{t}^{*}$. Employing the Hamilton\textendash Jacobi\textendash Bellman
equations we obtain nonlinear second order partial differential equations
for the optimally controlled value function.
\begin{thm}[Solution of the risk-averse Merton problem]
\label{thm:Merton}The PDE~\eqref{eq:HJB-2} has the explicit solution
\[
R(t,x)=\left(\frac{1+(\nu\,\epsilon-1)e^{-\nu(T-t)}}{\nu}\right)^{\gamma}\frac{x^{1-\gamma}}{1-\gamma},
\]
where $\nu:=-r\frac{1-\gamma}{\gamma}-\frac{1-\gamma}{\gamma^{2}}\left(\frac{\left((\mu-r)^{2}+s_{\rho}^{2}\sigma^{2}\right)}{2\sigma^{2}}-\frac{s_{\rho}}{\sigma}\right)$.
Moreover, the optimal controls are 
\begin{alignat*}{1}
\pi^{*} & =\max\left(\frac{(\mu-r)-s_{\rho}\sigma}{\sigma^{2}\gamma},0\right)\text{ and}\\
c_{t}^{*}(x) & =\frac{x\,\nu}{1+(\nu\,\epsilon-1)e^{-\nu(T-t)}}.
\end{alignat*}
\end{thm}

\begin{proof}
We recall the PDE~\eqref{eq:HJB-2},
\begin{align*}
0 & =R_{t}-\frac{\left((\mu-r)^{2}+s_{\rho}^{2}\sigma^{2}\right)R_{x}^{2}}{2\sigma^{2}R_{xx}}+\frac{s_{\rho}R_{x}\left|R_{x}\right|}{\sigma R_{xx}}+rxR_{x}+\frac{\gamma}{1-\gamma}\left(R_{x}\right)^{\frac{\gamma-1}{\gamma}},\\
R(T,x) & =\epsilon^{\gamma}\frac{x^{1-\gamma}}{1-\gamma},
\end{align*}
and choose the ansatz $R(t,x)=f(t)^{\gamma}\frac{x^{1-\gamma}}{1-\gamma}$.
In this case the partial derivatives are given by
\begin{align*}
R_{t} & =\left(\gamma f(t)^{\gamma-1}f^{\prime}(t)\right)\frac{x^{1-\gamma}}{1-\gamma},\\
R_{x} & =f(t)^{\gamma}x^{-\gamma},\\
R_{xx} & =-\gamma f(t)^{\gamma}x^{-\gamma-1}.
\end{align*}
The terminal condition for our Merton problem is $R(T,x)=\epsilon^{\gamma}\frac{x^{1-\gamma}}{1-\gamma}$
hence $f(T)=\epsilon>0$. Setting $C_{1}:=-\frac{\left((\mu-r)^{2}+s_{\rho}^{2}\sigma^{2}\right)}{2\sigma^{2}}$
and $C_{2}:=\frac{s_{\rho}}{\sigma}$ for ease of notation we substitute
the derivatives in the PDE~\eqref{eq:HJB-2} and obtain the following
ordinary differential equation for $f$; 
\begin{align}
f^{\prime}(t) & =f(t)\left(-r\frac{1-\gamma}{\gamma}+\frac{1-\gamma}{\gamma^{2}}\left(C_{1}+C_{2}f^{\gamma}\right)\right)-1.\label{eq:ODE1}
\end{align}
For $\nu$ as defined in Theorem~\ref{thm:Merton}, the general solution
of the ordinary differential equation~\eqref{eq:ODE1} is 
\[
f(t)=\frac{1+(\nu\epsilon-1)e^{-\nu(T-t)}}{\nu},
\]
which is positive. The optimal value function thus is
\begin{align*}
R(t,x) & =\left(\frac{1+(\nu\epsilon-1)e^{-\nu(T-t)}}{\nu}\right)^{\gamma}\frac{x^{1-\gamma}}{1-\gamma}.
\end{align*}
It follows that the the optimal control is $\pi_{t}^{*}=\max\left(\frac{(\mu-r)-s_{\rho}\sigma}{\sigma^{2}\gamma},0\right)$,
where the optimal consumption process is $c_{t}^{*}=\frac{x\nu}{1+(\nu\epsilon-1)e^{-\nu(T-t)}},$
which concludes the proof.
\end{proof}
The following Figure~\ref{fig:OC} illustrates the optimal consumption
$c^{*}$ as a function of the risk level $s_{\rho}$ for $\gamma=0.4$,
$r=0.01$, $\mu=0.1$, $\sigma=0.3$ and $\epsilon=0.1$. The time
horizon is $T=4$ and we consider the wealth $w_{0}=1$. Note that
$s_{\rho}$ can take only values smaller than $\frac{\mu-r}{\sigma}$
as otherwise $\pi^{*}<0$. 
\begin{figure}[H]
\centering{}\includegraphics[viewport=0bp 0bp 461bp 307bp,clip,scale=0.6]{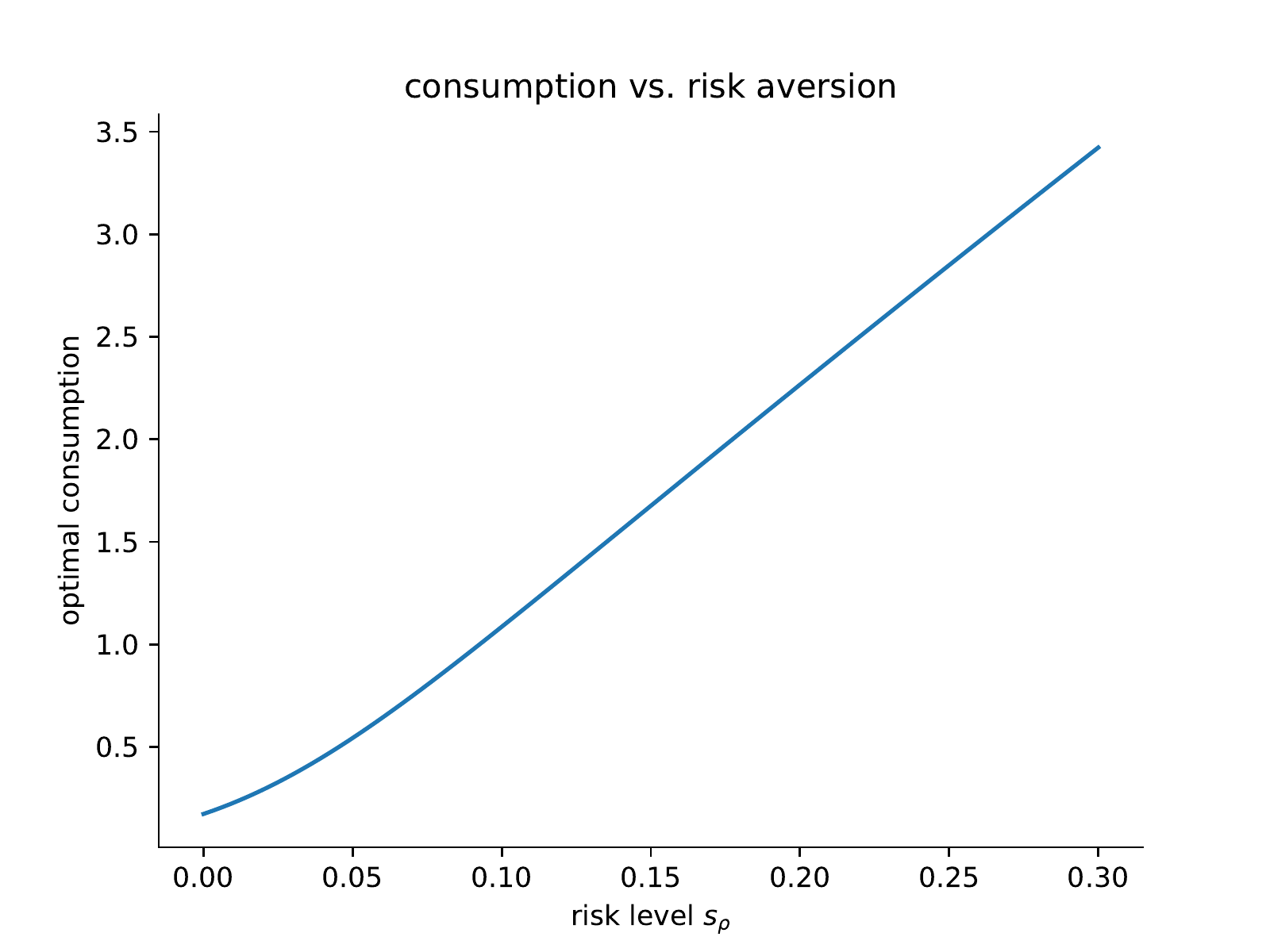}\caption{\label{fig:OC}optimal consumption}
\end{figure}

\section{Summary}

This paper introduces risk aversion in classical models of finance
by introducing nested risk measures. We demonstrate that classical
formulae, which are of outstanding importance in economics, are explicitly
available in the risk-averse setting as well. This includes the binomial
option pricing model, the Black\textendash Scholes model as well as
Merton's optimal consumption problem. 

We give an explicit Z-spread, which reflects the degree of risk aversion.
The Z-spread involves the volatility of the risky asset and a constant,
which indicates risk aversion. The results thus provide an economic
interpretation of the Z-spread by thorough risk management by iterating
risk measures.

We extend nested risk measures from a discrete time to a continuous
time setting. This allows deriving a non-linear risk generator expressing
the momentary dynamics of the classical model under risk aversion.
We demonstrate that the risk generator has a unique structure for
all coherent risk measures up to the constant, the coefficient of
risk aversion. The risk aversion constant is also naturally associated
with the Sharpe ratio. 

\section*{Acknowledgment}

We thank two anonymous referees for many insightful comments and helpful
discussions on this topic.

\bibliographystyle{abbrvnat}
\bibliography{LiteraturBook,LiteraturPaper}

\appendix

\section{A sufficient condition for Hölder continuity}

We give a sufficient condition for the Hölder continuity property~\eqref{eq:hoelder-continuity}
in Proposition~\ref{prop:RGen}. 
\begin{prop}
\label{prop:bla}Let $X$ be a solution of the stochastic differential
equation~\eqref{eq:SDE} where the drift $b$ and the diffusion $\sigma$
satisfy the \emph{usual conditions} of \citet[Theorem~5.2.1]{Oksendal2003}.
Furthermore suppose that for a fixed $p>2$ the moments $\E X_{t}^{p}$
are finite for every $t\in\mathcal{T}$ and that the diffusion coefficient
satisfies 
\[
\left|\sigma(t,x)-\sigma(s,x)\right|\leq\widetilde{D}\left|t-s\right|^{\alpha},\hfill\text{for all }x\in\mathbb{R}
\]
for some $\gamma\in(0,\nicefrac{1}{2})$. Then the Assumption~\eqref{eq:hoelder-continuity}
is satisfied, i.e., $\E C^{p}<\infty$ and in particular there exists
a constant such that
\[
\E C^{p}<C(\alpha,p,\mathcal{T},\widetilde{D}).
\]
\end{prop}

\begin{proof}
First observe that the usual conditions of \citet[Theorem~5.2.1]{Oksendal2003}
as well as the assumption in Proposition~\ref{prop:bla} ensures
that there exist $D,\widetilde{D}\in\mathbb{R}$ such that
\begin{align*}
\left|\sigma(t,X_{t})-\sigma(s,X_{s})\right| & \leq\left|\sigma(t,X_{t})-\sigma(t,X_{s})\right|+\left|\sigma(t,X_{s})-\sigma(s,X_{s})\right|\\
 & \leq D\left|X_{t}-X_{s}\right|+\widetilde{D}\left|t-s\right|.
\end{align*}
Therefore consider Hölder bounds on$\left|X_{t}-X_{s}\right|$, let
$p>2$ and recall the estimate
\[
(a+b)^{p}\leq2^{p}(a^{p}+b^{p}),\qquad a,b\geq0
\]
which implies 
\[
|X_{u}|^{p}\leq2^{p}\left|\int_{0}^{u}b(v,X_{v})\,\mathrm{d}v\right|^{p}+2^{p}\left|\int_{0}^{u}\sigma(v,X_{v})\,\mathrm{d}W_{v}\right|^{p}.
\]
Estimate the terms separately and consider the first term. Jensen's
inequality applied to the probability measure $\frac{\mathrm{d}v}{u}$
shows
\begin{align*}
\left|\int_{0}^{u}b(v,X_{v})\,\mathrm{d}v\right|^{p} & =u^{p}\left|\int_{0}^{u}b(v,X_{v})\,\frac{\mathrm{d}v}{u}\right|^{p}\\
 & \leq u^{p-1}\int_{0}^{u}|b(v,X_{v})|^{p}\,\mathrm{d}v\\
 & \leq C^{p}u^{p-1}\int_{0}^{u}(1+|X_{v}|)^{p}\,\mathrm{d}v\\
 & \leq2^{p}C^{p}u^{p-1}(u+\int_{0}^{u}|X_{v}|^{p}\,\mathrm{d}v).
\end{align*}
To estimate the second term use the Burkholder-Davis-Gundy inequality
\begin{align*}
\E\left|\int_{0}^{u}\sigma(v,X_{v})\,\mathrm{d}W_{v}\right|^{p} & \leq c(p)\E\left(\int_{0}^{u}|\sigma(v,X_{v})|^{2}\,\mathrm{d}v\right)^{\nicefrac{p}{2}}\\
 & \leq c(p)\E\left[u^{\nicefrac{p}{2}-1}\int_{0}^{u}\left|\sigma(v,X_{v})\right|^{p}\,\mathrm{d}v\right]\\
 & \leq c(p)C^{p}u^{\nicefrac{p}{2}-1}\E\int_{0}^{u}(1+\left|X_{v}\right|)^{p}\,\mathrm{d}v\\
 & \leq c(p)(2C)^{p}u^{\nicefrac{p}{2}-1}(u+\int_{0}^{u}\E\left|X_{v}\right|^{p}\,\mathrm{d}v).
\end{align*}
An application of Gronwall's lemma for both terms provides upper
bounds
\[
\E\left|\int_{s}^{t}b(v,X_{v})\,\mathrm{d}v\right|^{p}\leq C_{\text{Gronwall}}(C,p,t)(t-s)^{p}
\]
and 
\[
\E\left|\int_{s}^{t}\sigma(v,X_{v})\,\mathrm{d}W_{v}\right|^{p}\leq\widetilde{C}_{\text{Gronwall}}(C,p,t)(t-s)^{\frac{p}{2}}.
\]
It follows by adding up and choosing an appropriate constant $C^{*}$
that for $|t-s|<1$
\[
\E\left|X_{t}-X_{s}\right|^{p}\leq C^{*}(t-s)^{\frac{p}{2}}.
\]
As $p>2$ we can identify $\frac{p}{2}=1+\beta$ for some $\beta>0$
which shows that the assumptions of Kolmogorovs continuity theorem
(cf. Theorem 21.6 in \citet{Klenke2014}) are satisfied and implying
that there exists a random $C_{r}>0$ such that 
\begin{equation}
\left|X_{t}-X_{s}\right|\leq C_{r}\left|t-s\right|^{\alpha}\label{eq:1}
\end{equation}
for $\alpha\in(0,\frac{\beta}{p})$. 

It remains to show that the $p$-norm of $C_{r}$ in~\eqref{eq:1}
can be bounded. To show this, recall the Garsia\textendash Rodemich\textendash Rumsey
inequality. For any $p>1$ and $\delta>\frac{1}{p}$, there is a constant
$C(\delta,p)\in\mathbb{R}$ such that for any $f\in C([0,T],\mathbb{R})$
and $t,s\in[0,T]$ 
\[
\left|f(t)-f(s)\right|^{p}\le C(\delta,p)|t-s|^{\delta p-1}\int_{s}^{t}\int_{s}^{t}\frac{\left|f(u)-f(v)\right|^{p}}{|u-v|^{\delta p+1}}\,\mathrm{d}u\,\mathrm{d}v.
\]
It suffices to consider $C_{r}:=||X||_{\alpha;[0,T]}$, the Hölder
norm of $X$ defined by
\[
||X||_{\alpha;[0,T]}:=\sup_{0\le t<s\le T}\frac{\left|X_{t}-X_{s}\right|}{|t-s|^{\alpha}}.
\]
For any $0<\alpha<\frac{\beta}{p}$ take $\delta\in\left(\frac{1}{p},\alpha+\frac{1}{p}\right)$
to get 
\begin{align*}
\E\left[||X||_{\alpha;[0,T]}^{p}\right] & \le C(\delta,p)\int_{0}^{T}\int_{0}^{T}\frac{\E\left[\left|X_{t}-X_{s}\right|^{p}\right]}{|u-v|^{\delta p+1}}\,\mathrm{d}u\,\mathrm{d}v\\
 & \leq C(\delta,p)C^{*}\int_{0}^{T}\int_{0}^{T}\frac{\left|t-s\right|^{1+\beta}}{|u-v|^{\delta p+1}}\,\mathrm{d}u\,\mathrm{d}v\\
 & =C(\delta,p)C^{*}\int_{0}^{T}\int_{0}^{T}|u-v|^{\beta-\delta p}\,\mathrm{d}u\,\mathrm{d}v\\
 & =C(\delta,p,T)C^{*}.
\end{align*}
Here the second inequality follows from the first step. We conclude
the assertion by observing that 
\begin{align*}
\left|\sigma(t,X_{t})-\sigma(s,X_{s})\right| & \leq D\left|X_{t}-X_{s}\right|+\left|\sigma(t,X_{s})-\sigma(s,X_{s})\right|\\
 & \leq DC_{r}\left|t-s\right|^{\alpha}+\widetilde{D}\left|t-s\right|^{\alpha}\\
 & =(DC_{r}+\widetilde{D})\left|t-s\right|^{\alpha},
\end{align*}
where the constant $DC_{r}+\widetilde{D}$ is $p$-integrable.
\end{proof}

\end{document}